\documentclass[sort&compress,times,3p]{elsarticle}
\pdfoutput=1
\usepackage[utf8]{inputenc}
\usepackage{hyphenat}
\usepackage{todonotes}
\usepackage{paralist}
\usepackage[titlenumbered,vlined,linesnumbered,ruled]{algorithm2e}
\usepackage{mathtools}
\usepackage{tikz}
\usetikzlibrary{positioning}
\usepackage{pgfplots}
\pgfplotsset{compat=1.3}

\SetCommentSty{textrm}
\DontPrintSemicolon
\usepackage[colorlinks,allcolors=blue]{hyperref}
\usepackage{doi}

\usepackage{amsmath}
\usepackage{amsthm}
\usepackage{amssymb}
\usepackage{microtype}
\usepackage[sort&compress,nameinlink,noabbrev,capitalize]{cleveref}
\makeatletter
\let\orgdescriptionlabel\descriptionlabel
\renewcommand*{\descriptionlabel}[1]{%
  \let\orglabel\label
  \let\label\@gobble
  \phantomsection
  \edef\@currentlabel{#1}%
  \let\label\orglabel
  \orgdescriptionlabel{#1}%
}
\makeatother
\tikzstyle{edge} = [fill=gray,draw=gray,line cap=round, line join=round, line width=30pt]
\tikzstyle{vertex} = [fill=white, draw=black,shape=circle,node distance=80pt]

\pgfdeclarelayer{background}
\pgfsetlayers{background,main}

\SetKw{Break}{break}
\SetKw{Continue}{continue}

\newcommand{\Dom}{\text{Dom}}
\newcommand{\val}{\text{val}}
\newcommand{\poly}{\text{poly}}
\newcommand{\N}{\mathbb N}
\newcommand{\M}{\nu}
\newcommand{\IG}{\mathcal I}
\newcommand{\IE}{I^E}
\newcommand{\IV}{I^V}
\newcommand{\RE}{R^E}
\newcommand{\RV}{R^V}

\newcommand{\mat}{}

\newcommand{\PE}{P^E}
\newcommand{\PV}{P^V}

\newcommand{\dn}{\nabla}

\newcommand{\NCeins}{\ensuremath{\text{NC}^1}}

\newcommand{\algoCPU}{KernelCPU}
\newcommand{\algoGPU}{KernelGPU}
\newcommand{\algoGPUFE}{KernelGPU-FE}

\newtheorem{theorem}{Theorem}[section]
\newtheorem{problem}[theorem]{Problem}
\newtheorem{observation}[theorem]{Observation}
\newtheorem{lemma}[theorem]{Lemma}
\newtheorem{inv}[theorem]{Invariant}

\newtheorem{definition}[theorem]{Definition}
\newtheoremstyle{iostuff}%
{0pt}%
{0pt}%
{\hangindent=\parindent}%
{}%
{\itshape}%
{:}%
{.5em}%
{}%

\theoremstyle{iostuff}
\newtheorem*{probinstance}{Input}
\newtheorem*{probtask}{Question}

\crefname{theorem}{Theorem}{Theorems}
\crefname{problem}{Problem}{Problems}
\crefname{lemma}{Lemma}{Lemmas}
\crefname{example}{Example}{Examples}
\crefname{rrule}{Reduction Rule}{Reduction Rules}
\crefname{definition}{Definition}{Definitions}
\crefname{algorithm}{Algorithm}{Algorithms}

\begin{document}
\begin{frontmatter}
  \title{Serial and parallel kernelization of Multiple Hitting Set
    parameterized by the Dilworth number, implemented on the
    GPU\tnoteref{t1}} \tnotetext[t1]{The results presented in this
    work are unrelated to the authors' work at Huawei or Recraft,
    respectively.}

  \author[hw]{René van Bevern\corref{cor1}}%
  \ead{rene.van.bevern@huawei.com}
  \address[hw]{Huawei Technologies Co., Ltd.}
  \cortext[cor1]{Corresponding author.}

  \author[sfu]{Artem M.\ Kirilin}
  \ead{kirilinartem@gmail.com}
  \address[sfu]{Siberian Federal University, Krasnoyarsk, Russian Federation}

  \author[hw]{Daniel A.\ Skachkov}
  \ead{skachkov.daniel@huawei-partners.com}

  \author[recraft]{Pavel V.\ Smirnov}
  \address[recraft]{Recraft}
  \ead{pavel.vl.smirnov@recraft.ai}
  
  \author[hw]{Oxana Yu.\ Tsidulko}
  \ead{tsidulko.oksana@huawei.com}

\begin{abstract}
  The NP-hard
  Multiple Hitting Set problem is
  the problem of finding a minimum\hyp cardinality set
  intersecting each of the sets in a given
  input collection
  a given number of times.
    Generalizing
    a well\hyp known data reduction algorithm
    due to \citeauthor{Wei98},
    we show a problem kernel for Multiple Hitting Set
    parameterized by the Dilworth number,
    a graph parameter introduced
    by \citeauthor{FH78} in \citeyear{FH78}
    yet
    seemingly so far unexplored in the context of parameterized
    complexity theory.
    Using matrix multiplication,
    we speed up the algorithm
    to quadratic sequential time
    and logarithmic parallel time.
    We experimentally evaluate
    our algorithms.
    By implementing our algorithm on GPUs,
    we show the feasibility
    of realizing kernelization algorithms
    on SIMD (Single Instruction, Multiple Data) architectures.
\end{abstract}
\end{frontmatter}

\section{Introduction}
\thispagestyle{empty}
\noindent
The following fundamental
combinatorial optimization problem
arises in bioinformatics \citep{MK13},
medicine \citep{Vaz09,MPMM10},
clustering \citep{HKMN10,BMN12},
automatic reasoning  \citep{Rei87,FBB18,BTU19},
feature selection \citep{FBNS16,BFL+19},
radio frequency allocation \citep{SMNW14},
software engineering \citep{OC03},
and public transport optimization \citep{Wei98,BFFS19}.

\begin{problem}[Multiple Hitting Set]
  \begin{probinstance}
    A hypergraph~$H=(V,E)$
    with \emph{vertices}~$V=\{v_1,\dots,v_n\}$ and
    \emph{hyperedges}~$E=\{e_1,\dots,e_m\}\subseteq 2^V$,
    a \emph{demand} function~$f:E\rightarrow\{1,\dots,\alpha\}$,
    and $k,\alpha\in\mathbb N$.
  \end{probinstance}
  \begin{probtask}
    Is there a \emph{hitting set}~$S\subseteq V$
    of cardinality at most~$k$
    such that $|e\cap S|\geq f(e)$
    for each~$e\in E$?
  \end{probtask}
\end{problem}
\noindent
Already the special case
with $f\equiv 1$, known as simply
\emph{Hitting Set},
is NP\hyp complete \citep{Kar72}.
Exact algorithms for NP\hyp complete problems
usually take time exponential in the input size.
Thus,
an important preprocessing step is data reduction,
which has proven to significantly
shrink real\hyp world instances of NP\hyp hard problems
\citep{Wei98,MPMM10,BFFS19,BFT20b,BS20,BMST23,ALM+22,ABG+20}.
In the context
of public transport optimization,
\citet{Wei98} introduced a simple but
very effective in experiments \citep{Wei98,BFFS19,BS20}
data reduction algorithm for Hitting Set.
It
exhaustively applies two data reduction rules
that do not alter the answer
to the input instance:

\begin{description}
\item[(W1)\label{w1}] If there is a pair of
  distinct vertices~$v_i$ and~$v_j$
  such that every hyperedge containing $v_j$
  also contains~$v_i$,
  delete~$v_j$.
\item[(W2)\label{w2}] If there is a pair of
  distinct hyperedges~$e_i$ and $e_j$
  such that $e_i\subseteq e_j$, then delete~$e_j$.
\end{description}
\noindent
\looseness=-1
Using \emph{kernelization}
from \emph{parameterized complexity theory},
which is formally defined in \cref{sec:prelim},
our work contributes to the understanding,
generalizes and speeds up \citeauthor{Wei98}'s
data reduction algorithm in the following ways.

\paragraph{Understanding}
While the data reduction effect of \citeauthor{Wei98}'s
algorithm is experimentally proven \citep{BFFS19,Wei98,BS20},
finding theoretical explanations for its effectivity
is challenging \citep{BFFS19,Fel02}.
For example, 
\citet{Fel02} asks
``Weihe's example looks like an FPT kernelization,
but what is the parameter?''
We show that \citeauthor{Wei98}'s algorithm
actually computes problem kernels for Hitting Set
whose number of vertices and hyperedges is linear in the
\emph{Dilworth number} of the incidence graph of
the input hypergraph
(see \cref{sec:dn} for a definition).

The Dilworth number has been introduced by \citet{FH78}
in \citeyear{FH78},
which led to a series of studies of the structure
of graphs with small Dilworth number
\citep[among others]{CP14,FRS03,HM89}.
However,
the Dilworth number until now seems unexplored in
a parameterized complexity context.
For example,
neither
\citet{GHHH18} nor \citet{SW19}
list it.
This is surprising,
since the Dilworth number is bounded from above by
the neighborhood diversity (see \cref{sec:dn}),
which is a
well\hyp established parameter
in parameterized complexity studies
\citep{Lam12,Gan12,GR15,AAK+20,GKK20},
and it seems a logical step
to analyze which parameterized algorithms
for the neighborhood diversity
can be strengthened to use the Dilworth number instead.

We note here that
the Dilworth number of the incidence graph
of a hypergraph~$H=(V,E)$ can be exponential in
the number~$|V|$ of vertices;
however,
one \emph{in principle} cannot expect
problem kernels of size polynomial in~$|V|$
for Hitting Set
unless the polynomial\hyp time hierarchy
collapses \citep{DLS14,DM14}.
Indeed,
for problem kernels of size polynomial in~$|V|$ to exist,
the cardinality of each input hyperedge must be bounded from above by a constant~$d$ \citep{DM14}.
In this case,
problem kernels of size~$k^{O(d)}$
have been shown
\citep{NR03,Dam06,FG06,Abu10,Mos10,Kra12,FK14,BST20,BHRT22,BS20,AMW20}.

\paragraph{Generalization}
\looseness=-1
Motivated by problems in
feature selection,
optimal cancer medication,
and genome-wide association studies,
attempts have been undertaken to generalize
\citeauthor{Wei98}'s
data reduction algorithm to Multiple Hitting Set
\citep{CSM04,MPMM10,MMM+16}.
However,
\citet{CSM04}
only generalize the hyperedge deletion rule \ref{w2}.
The generalization of the vertex deletion rule \ref{w1}
of \citet{MMMB05}
is wrong:
as shown in \ref{apx:counterexample},
it \emph{does} alter the answer to the input instance.

We provide a full generalization
of \citeauthor{Wei98}'s algorithm in the sense
that we obtain a problem kernel for Multiple Hitting Set
whose number of vertices and hyperedges is
linear in the Dilworth number
of the incidence graph of the input hypergraph
and in the maximum hyperedge demand~$\alpha$.

We provide additional data reduction rules
that, in the case of Multiple Hitting Set
and in contrast to \ref{w2},
allow for safe deletion even of hyperedges that
are \emph{not} contained within each other.
While not provably lowering the size of problem kernels,
we show their significant additional data reduction effect in experiments.

\paragraph{Speed-up}
\looseness=-1
A comparison between
\citeauthor{Wei98}'s algorithm and
linear\hyp time data reduction algorithms for Hitting Set  has shown that
the data reduction effect of \citeauthor{Wei98}'s algorithm
is clearly superior when there are large hyperedges,
yet the algorithm is significantly slower,
even so much so as
to make it inapplicable to large hypergraphs \citep{BS20}.
By looking at the algorithm
through the lens of multiplying
incidence matrices of the input hypergraph,
we provide quadratic sequential\hyp time
and logarithmic parallel\hyp time variants
of our kernelization algorithms,
thus contributing to
parallel kernelization studies,
which have recently gained increased interest \citep{BST20,BT20}.
In contrast to previous,
purely theoretic and proof\hyp of\hyp concept
work on parallel kernelization,
we implement our algorithm on GPUs
and thus prove the feasibility
of realizing kernelization algorithms on 
SIMD (Single Instruction, Multiple Data) 
architectures,
in which all cores of a multiprocessor execute the same operation,
yet on different (parts of) the data.
Kernelization algorithms,
which are often a set of data reduction rules applied
on different parts of the input,
seem to lend themselves to implementation on SIMD architectures.

\paragraph{Organization of this work}
\cref{sec:prelim}
introduces basic graph\hyp theoretic,
parameterized complexity,
and kernelization terminology.
\cref{sec:kern}
presents
known and new
data reduction rules for Multiple Hitting Set,
and shows that they yield
a problem kernel whose number of vertices and hyperedges is linear
in the Dilworth number.
\cref{sec:matmult}
shows how to obtain
fast serial and parallel implementations
of the data reduction rules
via matrix multiplication.
\cref{sec:exp} experimentally evaluates the effect
and speed of our data reduction rules.

\section{Preliminaries}
\label{sec:prelim}
\noindent
In this section,
we introduce basic graph\hyp theoretic,
parameterized complexity,
and kernelization terminology.

\subsection{Graphs and hypergraphs}
\paragraph{Hypergraphs}
A \emph{hypergraph} is a pair $H = (V, E)$,
where $V$~is a set of \emph{vertices} and~$E\subseteq 2^{V}$ is a set of \emph{(hyper)edges}.
We will often denote the vertices
and hyperedges as~$V=\{v_1,\dots,v_n\}$ and $E=\{e_1,\dots,e_m\}$,
respectively.

For a vertex~$v\in V$,
by $E(v):=\{e\in E\mid v\in e\}$,
we denote the set of hyperedges containing~$v$.
By $|H| = |V| + \sum_{e\in E}|e|$,
we denote the \emph{size} of the hypergraph.
This notion of size
is motivated by the fact
that each incidence between a vertex and a hyperedge
has to be encoded in some form.
The \emph{incidence matrix}~$A(H)$ of a hypergraph~$H=(V, E)$
 is an
$(m\times n)$-matrix such that
\[
    A_{ij} = \begin{cases}
        1,&\text{if $v_j\in e_i$ and}\\
        0,&\text{otherwise.}
    \end{cases}
\]

\paragraph{Undirected graphs}
A hypergraph whose hyperedges have cardinality two
is an \emph{(undirected) graph}.
For a vertex~$v\in V$ of a graph~$G=(V,E)$,
$N(v):=\{u\in V\mid \{u,v\}\in E\}$
denotes the \emph{open neighborhood of~$v$}
and $N[v]:=N(v)\cup\{v\}$ denotes the
\emph{closed neighborhood}.
To avoid confusion,
throughout this work,
the notation $N(v)$ is exclusively applied to graphs,
whereas $E(v)$ is exclusively applied to hypergraphs (that are not graphs).

A \emph{matching} in a graph
is a set of pairwise disjoint edges.
The \emph{matching number $\M(G)$}
of a graph~$G$ is the maximum cardinality
of any matching in~$G$.
The \emph{incidence graph}~$\IG(H)$ of
a hypergraph~$H=(V,E)$
is a
bipartite graph with the vertex set~$V\cup E$
and %
an edge~$\{v,e\}$ for each~$v\in V$ and~$e\in E$
such that $v\in e$.
That is,
for any hyperedge~$e$ of~$H$,
we have $N(e)=e$ in the incidence graph~$\IG(H)$.

\paragraph{Directed graphs}
A \emph{directed graph}~$G=(V,A)$
consists of vertices~$V=\{v_1,\dots,v_n\}$
and \emph{arcs}~$A=\{a_1,\dots,a_m\}\subseteq V^2$.
If $(v_i,v_j)\in A$,
we call $v_i$ a \emph{parent} of $v_j$.
We call $v_i$ an \emph{ancestor} of~$v_j$
if $v_i=v_j$, if $v_i$~is a parent of~$v_j$,
or if $v_i$~is a parent of an ancestor of~$v_j$.
A vertex without parents is called a \emph{source},
a vertex that is not a parent of any other vertex is a \emph{sink}.

A \emph{cycle} in a directed graph is
a sequence of vertices~$v_1,v_2,\dots,v_\ell$
such that $(v_i,v_{i+1})\in A$ for $i\in\{1,\dots,\ell-1\}$
and $(v_\ell,v_1)\in A$.
A directed graph is \emph{acyclic} if it contains no cycles.
\begin{observation}
  \label{obs:alive-ancestors}
  Let~$G=(V,A)$ be a directed acyclic graph and $X\subseteq V$.
  Every $v\in X$ has
  an ancestor~$u$ in~$X$,
  all of whose parents are not in~$X$
  (possibly, $u=v$
  and that set of parents may be empty).
\end{observation}

\subsection{Complexity theory and kernelization}
\noindent
Formally,
we study decision problems~$\Pi\subseteq\{0,1\}^*$,
where the task is to decide whether a given~$x\in\{0,1\}^*$
belongs to~$\Pi$.

\paragraph{Parameterized complexity}
A \emph{parameterized problem}
is a pair
$(\Pi,\kappa)$
where $\Pi\subseteq\{0,1\}^*$ is a decision problem
and $\kappa\colon\{0,1\}^*\to\N$ is 
a polynomial\hyp time computable function
called a \emph{parameterization}.

A~\emph{kernelization} for a parameterized problem~$(\Pi,\kappa)$
is a polynomial\hyp time algorithm
that maps any instance~$x\in\{0,1\}^*$
to an instance~$x'\in\{0,1\}^*$
such that $x\in \Pi \iff x'\in \Pi$ and such that
$|x'|\leq g(\kappa(x))$ for some computable function~\(g\).
We call \(x'\) the \emph{problem kernel}
and \(g\) its \emph{size}.

\paragraph{Circuit families}
A \emph{Boolean circuit}
with $n$~inputs and $m$~outputs
is a directed acyclic graph
with $n$~sources and $m$~sinks;
each of its
its non\hyp source vertices~$v$
has one of the following three types:
\begin{itemize}[---]
\item $v$~is labeled ``$\neg$'' and has exactly one parent,
  
\item $v$~is labeled ``$\vee$'' and has exactly two parents,
  
\item $v$~is labeled ``$\wedge$'' and has exactly two parents.
\end{itemize}
Identifying 1 with the Boolean value \emph{true} and
0 with the Boolean value \emph{false},
we can inductively define the
\emph{value}~$\val(v)$ of each vertex~$v$ of the Boolean circuit
on \emph{input} $x\in\{0,1\}^n$ as follows.
\begin{itemize}[---]
\item The value of the $i$-th source is~$x_i$,
\item For any vertex~$v$ labeled ``$\neg{}$'' with parent $u$,
  $\val(v):=\neg\val(u)$.
\item For any vertex~$v$ labeled ``$\vee$'' with parents~$u_1$ and~$u_2$,
  $\val(v):=\val(u_1)\vee\val(u_2)$.
\item For any vertex~$v$ labeled ``$\wedge$'' with parents~$u_1$ and~$u_2$,
  $\val(v):=\val(u_1)\wedge\val(u_2)$.
\end{itemize}
Denoting the sinks as $u_1,\dots,u_m$,
we call $\val(u_1)\val(u_2)\dots\val(u_m)\in\{0,1\}^m$
the \emph{output} of the Boolean circuit on input~$x$.
The \emph{size} of a Boolean circuit is its number of vertices.
The \emph{depth} of a Boolean circuit is the length
of the longest path from any source to any sink.

An \emph{$\NCeins{}$ circuit family} is a
sequence~$(C_n)_{n\in \N}$ of Boolean circuits,
each with
$n$~inputs, $O(\log n)$ depth, and $\poly(n)$ size.
We say that an \NCeins{} circuit family~$(C_n)_{n\in\N}$
decides a problem~$\Pi\subseteq\{0,1\}^*$
if $x\in \Pi$ if and only if the circuit~$C_{|x|}$ outputs 1
when given~$x$ on the input.
More generally,
we say that an \NCeins{} circuit family~$(C_n)_{n\in\N}$
computes a function~$f\colon\{0,1\}^*\to \{0,1\}^*$
if,
for any~$x\in\{0,1\}^*$,
the circuit $C_{|x|}$ outputs $f(x)$ on input~$x$.

The outputs of \NCeins{} circuits can be computed
in logarithmic time on a parallel computer with
a polynomial number of processors,
where each processor computes the value of one vertex,
taking as input the values of its parents
\citep[Section~6.7]{AB09}.

\paragraph{Input encoding}
For our parallel hypergraph algorithms (or \NCeins{} circuits),
we will generally assume incidence matrices as input,
so that each parallel processor gets one matrix entry as input.
For sequential hypergraph algorithms,
we assume a list of vertices
and a list of hyperedges on the input.

\subsection{Dilworth number and neighborhood diversity}
\label{sec:dn}
\noindent
Consider the relation~$\sqsubseteq$
on the vertices of a graph~$G=(V,E)$ such that
\[
  u\sqsubseteq v\iff N(u)\subseteq N[v].
\]
This relation is reflexive and transitive
and is called the \emph{vinical order} of~$G$ \citep{FH78}.
We call two vertices~$u$ and~$v$ \emph{incomparable}
if $u\not\sqsubseteq v$ and~$v\not\sqsubseteq u$.
A \emph{chain} is a subset of pairwise comparable vertices,
whereas an \emph{antichain} is a subset
of pairwise incomparable vertices.
The \emph{Dilworth number~$\dn(G)$} of~$G$
is the size of a largest antichain in~$\sqsubseteq$~\citep{FH78},
which,
by Dilworth's theorem \citep{Dil50},
is equivalent to
the minimum number of chains whose union is~$V$.

A related frequently studied graph parameter
is the neighborhood diversity
\citep{Lam12,Gan12,GR15,AAK+20,GKK20}.
To introduce it,
consider the relation~$\sim$ on the vertices
of a graph~$G=(V,E)$ such that
\[
  u\sim v\iff N(u)\setminus\{v\}=N(v)\setminus\{u\}.
\]
This is an equivalence relation \citep{Lam12}
and the \emph{neighborhood diversity~$\delta(G)$}
of~$G$ is the number of equivalence classes of~$\sim$.

\citet{FH78} relate the Dilworth number
to many other graph parameters,
yet not to the neighborhood diversity,
which is frequently used for parameterized
complexity analysis but was introduced much later \citep{Lam12}.
It is not hard to show
that the Dilworth number is upper-bounded by the neighborhood diversity,
but the gap between the two can be arbitrarily large:

\begin{lemma}\leavevmode
  \begin{compactenum}[(i)]
  \item\label{dD1} For any graph~$G$, one has   $\dn(G)\leq\delta(G)$.
    
  \item\label{dD2} For any~$n\in\N$,
    there is a graph~$G$ on $2n$~vertices with
    $1=\dn(G)\leq \delta(G)=2n-1$.
  \end{compactenum}
\end{lemma}

\begin{proof}
  \eqref{dD1}
  Consider any pair of vertices~$u$ and~$v$ of~$G$.
  Then,
  \begin{align*}
u\sim v\implies&&    N(u)\setminus \{v\}&\subseteq N(v)\setminus\{u\}\\
    \implies &&    N(u)\setminus\{v\}&\subseteq N(v)\\
    \implies &&    N(u)&\subseteq N(v)\cup\{v\}\\
    \implies && u\sqsubseteq v.
  \end{align*}
  By definition,
  the vinical order~$\sqsubseteq$ of~$G$
  has an antichain~$C$ of cardinality~$\dn(G)$.
  For any pair of distinct vertices~$u$ and~$v$ in~$C$,
  we have $u\not\sqsubseteq v$ and thus $u\not\sim v$.
  It follows that $\sim$ has at least~$|C|$ equivalence classes,
  and thus
  $\dn(G)=|C|\leq\delta(G)$.

  \eqref{dD2}
  Construct a graph~$G$ on $2n$~vertices as follows.
  Start with an empty graph.
  Then,
  for each $i\in\{1,\dots,n\}$,
  add first an isolated vertex~$u_i$
  and then add a vertex~$v_i$ that is adjacent to all previously
  added vertices.
  We get $\delta(G)=2n-1$ since each pair of vertices in~$G$,
  except for $u_1$ and~$v_1$,
  are pairwise nonequivalent under~$\sim$:
  \begin{itemize}[---]
  \item For any~$i,j\in\{1,\dots,n\}$ such that $i\ne j$,
    one has $u_i\not\sim v_j$ 
    since $v_j$ is adjacent to~$u_j$ but $u_i$~is~not
    
  \item For any~$i\in\{2,\dots,n\}$,
    one has $u_i\not\sim v_i$,
    since $v_i$~is adjacent to~$v_1$ but $u_i$ is not.
  \item 
  For $1\leq i<j\leq n$,
  one has
    $v_i\not\sim v_j$ since
    $v_j$ is adjacent to~$u_j$, but $v_i$ is nonadjacent to~$u_j$, and one has
    $u_i\not\sim u_j$
    since $v_i$ is adjacent to~$u_i$ but $u_j$~is nonadjacent to~$v_i$.
  \end{itemize}
  We also get $\dn(G)=1$
  since
  each pair of vertices is comparable
  in the vinical order of~$G$:
  for $1\leq i\leq j\leq n$,
  one has
  \begin{itemize}
  \item[{$N(v_i)\subseteq N[v_j]$}] since
    any vertex added to~$G$
    before~$v_j$ is obviously in~$N[v_j]$
    and any vertex added to~$G$
    after~$v_j$ is adjacent
    either to both of~$v_i$ and~$v_j$
    or to none of them,
  \item[{$N(u_j)\subseteq N[u_i]$}] since
    any vertex added to~$G$ before $u_j$
    is nonadjacent to~$u_j$
    and any vertex added to~$G$ after~$u_j$
    is adjacent either to all or none of~$u_i$ and~$u_j$.
  \end{itemize}
  Finally,
  one has $N(u_1)\subseteq N[v_1]$,
  that is, one can cover all vertices
  by the single chain
  $u_n\sqsubseteq u_{n-1}\sqsubseteq\dots\sqsubseteq u_1\sqsubseteq v_1\sqsubseteq v_2\sqsubseteq\dots\sqsubseteq v_n$.
  Thus,
  $\dn(G)=1$.
\end{proof}

\section{A problem kernel for Multiple Hitting Set}
\label{sec:kern}
\noindent
In \cref{sec:dr},
we present several data reduction rules for Multiple Hitting Set
that generalize \citeauthor{Wei98}'s rules \ref{w1} and~\ref{w2}.
In \cref{sec:kern2},
we show how problem kernels can be obtained by applying (subsets of)
these data reduction rules.
Later,
in \cref{sec:exp},
we experimentally evaluate
several combinations of these data reduction rules.

\subsection{Data reduction rules}
\label{sec:dr}

\noindent
The following two data reduction rules
were
suggested by \citet{CSM04}.
The ``full edge'' rule \ref{rule:full-edges} exploits that all vertices in a
hyperedge~$e_j$ with demand~$f(e_j)=|e_j|$ must belong to any feasible
solution:

\begin{description}
\item[(FE)\label{rule:full-edges}\label{FE}]
  If there is a hyperedge~$e_j\in E$
  such that $|e_j|=f(e_j)$,
  then delete~$e_j$,
  delete each $v\in e_j$,
  decrement $k$ by~$|e_j|$,
  and decrement $f(e_i)$
  by one for each hyperedge~$e_i$ containing~$v$,
  deleting hyperedges~$e_i$ for which $f(e_i)$ reaches 0.
\end{description}
The ``superedge'' rule \ref{rule:superedges} is a straightforward
generalization of \ref{w2} %
and exploits that
the deleted hyperedge~$e_j$ will be hit whenever $e_i$ is hit:
\begin{description}
\item[(SE)\label{rule:superedges}\label{SE}]
  If there is a pair of distinct
  hyperedges $e_i, e_j \in E$
  such that $e_i \subseteq e_j$ and $f(e_i) \geq f(e_j)$,
  then
  delete~$e_j$.
\end{description}

\noindent
Interestingly,
one can further generalize \ref{SE} so that it may delete even
hyperedges that are \emph{not} contained one in another.
Assume,
for example,
two distinct yet intersecting hyperedges~$e_i,e_j\in E$.
Any hitting set~$S$ has to contain $f(e_i)$~vertices of~$e_i$,
yet
there are only $|e_i\setminus e_j|$ elements in~$e_i$ that are not
also in~$e_j$.
Thus,
the remaining $f(e_i)-|e_i\setminus e_j|$ elements of~$S$ must be in~$e_i\cap e_j$,
we say that $e_i$ \emph{pushes} this amount of demand to~$e_j$.
If $e_i$ pushes at least $f(e_j)$~units of demand to~$e_j$,
then we know that $e_j$ will be hit whenever $e_i$~is,
and can thus delete~$e_j$.
\begin{definition}
  A hyperedge~$e_i$ \emph{supersedes} a hyperedge~$e_j$
  if $f(e_i)-|e_i\setminus e_j|\geq f(e_j)$.
\end{definition}
\noindent
This gives rise to the following ``demand pushing'' rule \ref{DP},
which subsumes rule \ref{rule:superedges} of \citet{CSM04}.
\begin{description}
\item[(DP)\label{rule:pushing-demand}\label{DP}]
  If there is a pair of distinct
  hyperedges~$e_i,e_j\in E$
  such that $e_i$ supersedes~$e_j$,
  then delete~$e_j$.
\end{description}

\noindent
To further generalize the data reduction rule,
one can exploit that,
if a hyperedge~$e_j\in E$ intersects several hyperedges,
then every hyperedge~$e_i$ intersecting~$e_j$
pushes some demand to~$e_i\cap e_j$.
If satisfying these demands requires
at least $f(e_j)$~elements from~$e_j$,
then one can delete~$e_j$.
This leads to the following data reduction rule.
\begin{description}
\item[(LP)\label{rul:lp}\label{LP}]
  For a hyperedge~$e_j\in E$,
  consider
  the hypergraph~$H_j$
  on the same vertex set as~$H$
  yet
  containing, for each hyperedge~$e_i\in E$,
  a hyperedge~$e_i\cap e_j$ with demand $f(e_i)-|e_i\setminus e_j|$
  whenever this value is positive.
  Now,
  consider a lower bound~$L_j$
  on the cardinality of any multiple hitting set for~$H_j$.%
  \footnote{The lower bound $L_j$ can be
obtained, for example,
via an integer linear programming relaxation.}
  If $L_j\geq f(e_j)$,
  then delete~$e_j$.
\end{description}
\noindent
\looseness=-1
To verify the correctness of \ref{rul:lp},
observe that,
even if the multiple hitting set in~$H$
contains $e_i\setminus e_j$
for each $e_i\in E$,
it still needs to contain at least $f(e_i)-|e_i\setminus e_j|$ from~$e_i\cap e_j$.
If meeting these demands for all~$e_i\cap e_j$
requires at least~$f(e_j)$ vertices,
that is, if the rule's condition is satisfied,
then~$e_j$ will be hit $f(e_j)$~times anyway,
since $e_j\supseteq e_i\cap e_j$ for each $e_i\in E$.

Finally,
to prove a problem kernel,
we will also provide a generalization of \citeauthor{Wei98}'s data reduction rule \ref{w1}
to Multiple Hitting Set.
Unfortunately,
a previous generalization attempt of \citet{MMMB05}
turned out to be wrong (see \ref{apx:counterexample}).
To state the data reduction rule,
we need the following definition.

\begin{definition}
  For a pair of distinct vertices $v_i, v_j \in V$,
  we say that~$v_i$ \emph{dominates (or is a dominator for)}~$v_j$
  if $E(v_j) \subseteq E(v_i)$.
  By $\Dom(v_j)$, we denote all dominators for~$v_j$.
\end{definition}

\noindent
Note that,
if $v_i$~dominates $v_j$,
then $v_j$~can be replaced by~$v_i$ in any hitting set.
Thus,
if
$|\Dom(v_j)|\geq\max_{e \in E(v_j)}f(e)$,
then we can safely delete~$v_j$
from the hypergraph,
since $\Dom(v_j)$~contains a sufficient amount of vertices
to satisfy the demand of any hyperedge containing~$v_j$.
This gives rise to the following ``multiple domination'' rule.
\begin{description}
\item[(MD)\label{rule:multiple-domination}\label{MD}]
  If there is a vertex~$v_j$ such that $|\Dom(v_j)|\geq\max_{e \in E(v_j)}f(e)$,
  then delete~$v_j$.
\end{description}

\subsection{Problem kernel size}
\label{sec:kern2}

\noindent
In this section,
we show how to use the data reduction rules
presented in \cref{sec:dr}
to obtain a problem kernel for Multiple Hitting Set.

For the kernel size analysis,
we will use the following lemma,
which relates vertex dominance and hyperedge inclusion
within a hypergraph to
comparability in the vinical order of the incidence graph.

\begin{lemma}
  \label{lem:inc}
  Let $H$~be a hypergraph and $\sqsubseteq$~be the vinical order
  of its incidence graph~$\IG(H)$.  Then
  \begin{enumerate}[(i)]
  \item\label{inc1} if, for two hyperedges~$e_i$ and~$e_j$,
    one has $e_i\nsubseteq e_j$ and $e_j\nsubseteq e_i$,
    then $e_i$ and~$e_j$ are incomparable with respect to~$\sqsubseteq$,
    
  \item\label{inc2} if, out of two vertices $v_i$ and $v_j$,
    neither dominates the other,
    then $v_i$ and $v_j$ are incomparable with respect to~$\sqsubseteq$,

  \end{enumerate}
\end{lemma}
\begin{proof}
  \eqref{inc1} Note that,
  in the incidence graph~$\IG(H)$,
  one has $N(e_i)=e_i$
  and $N(e_j)=e_j$.
  Towards a contradiction,
  assume $e_i\sqsubseteq e_j$.
  Then,
  by the definition of the vinical order~$\sqsubseteq$,
  one has
  $e_i=N(e_i)\subseteq N[e_j]=e_j\cup \{e_j\}$.
  Since hyperedges do not contain
  other hyperedges as elements,
  we definitively have $e_j\notin e_i$,
  and thus, in fact, $e_i\subseteq e_j$,
  contradicting the assumption that $e_i\nsubseteq e_j$.
  Analogously,
  from $e_j\sqsubseteq e_i$ follows
  the contradiction $e_j\subseteq e_i$.

  \eqref{inc2}.
  Note that,
  in the incidence graph~$\IG(H)$,
  one has $N(v_i)=E(v_i)$ and $N(v_j)=E(v_j)$.
  Towards a contradiction,
  assume that $v_j\sqsubseteq v_i$,
  that is, $E(v_j)=N(v_j)\subseteq N[v_i]=E(v_i)\cup\{v_i\}$.
  Since definitively $v_i\notin E(v_j)$,
  we in fact have $E(v_j)\subseteq E(v_i)$,
  contradicting the assumption that $v_i$ does not dominate~$v_j$.
  Analogously,
  the assumption that $v_i\sqsubseteq v_j$
  contradicts the assumption that $v_j$ does not dominate~$v_i$.
\end{proof}

\begin{theorem}\label{theorem:big-kernel}
    Given an instance of Multiple Hitting Set
    $H=(V,E)$ with $f:E\rightarrow\{1,\dots,\alpha\}$,
    a problem kernel $H'=(V',E')$
    with $|V'|+|E'|\leq2\alpha\dn(\IG(H))$
    and hyperedge demands $f'=f$
    can be computed as follows:
    \begin{enumerate}[Step 1.]
    \item Apply \ref{rule:superedges} as long as possible,
    \item Apply \ref{rule:multiple-domination} as long as possible.
    \end{enumerate}
\end{theorem}
\begin{proof}
  Let $H^*=(V,E^*)$~be the hypergraph obtained from~$H$
  by exhaustive application of \ref{SE},
  that is,
  $H^*$ is $H$ after Step 1.
  Consider $E^*_i:=\{e\in E^*:f(e)=i\}$
  for $i\in\{1,\dots,\alpha\}$. 
  Then, $\bigcup_{i=1}^{\alpha}E^*_i=E^*$.
  By the pigeonhole principle,
  $|E^*_{i^*}|\geq|E^*|/\alpha$
  for some~$i^*\in\{1,\dots,\alpha\}$.
  Since the hyperedges in $e\in E^*_{i^*}$
  have equal demand and survived \ref{SE},
  they are not contained in each other,
  neither in~$H^*$ nor in~$H$.
  Thus,
  by \cref{lem:inc}\eqref{inc1},
  $\dn(\IG(H))\geq |E^*_{i^*}|\geq |E^*|/\alpha\geq |E'|/\alpha$,
  since the second step does not increase the number of hyperedges.

  We have shown $|E'|\leq\alpha\dn(\IG(H))$.
  It remains to show~$|V'|\leq\alpha\dn(\IG(H))$.
  To this end,
  let $D\subseteq V'$
  be of maximal cardinality such
  that no vertex in~$D$ dominates
  any other vertex in~$V'$.
  Then, we can write
  \[
    V'=D\cup\bigcup_{u\in D}\Dom(u).
  \]
  Moreover,
  since \ref{MD} is inapplicable to~$H'$,
  we have $|\Dom(u)|\leq\alpha-1$ for any~$u\in D$.
  Also,
  by \cref{lem:inc}\eqref{inc2},
  the vertices in~$D$ are incomparable in
  the vinical order of~$\IG(H')$.
  Therefore,
  $|V'|\leq \alpha|D|\leq\alpha\dn(\IG(H'))\leq \alpha\dn(\IG(H))$.
\end{proof}

\noindent
We point out that there are graphs
for which the analysis provided in \cref{theorem:big-kernel}
is tight:
consider the hypergraph~$H=(V,E)$ with $n$~vertices $\{v_1,\dots,v_n\}$
and $n$~hyperedges of the form $\{v_i\}$ for~$i\in\{1,\dots,n\}$.
Its incidence graph is a disjoint union of $n$~copies of~$K_2$.
Its Dilworth number is~$n$,
none of \ref{SE} and \ref{MD} is applicable,
and one has $|V|+|E|=2n=2\alpha\dn(\IG(H))$.

\section{Efficient implementation using matrix multiplication}
\label{sec:matmult}
\noindent
\looseness=-1
In this section,
we show efficient parallel and sequential
implementations of \ref{MD} and \ref{DP},
the latter of which supersedes \ref{SE}.
The algorithms presented in this section
thus yield problem kernels for Multiple Hitting Set
via \cref{theorem:big-kernel}.

Actually, they go even further:
\cref{theorem:big-kernel}
tells us a two\hyp step recipe for computing
problem kernels for Multiple Hitting Set.
However,
after Step 2,
Step 1 may become applicable again
and
shrink the input hypergraph further.
The algorithms presented in this section
repeat the two steps until full exhaustion.
Interestingly,
the asymptotic running time of our sequential algorithm
will be the same %
regardless of whether we apply
the two steps in \cref{theorem:big-kernel} once or until exhaustion.

The key observation to both algorithms is the following.
Consider the $(m\times n)$-incidence matrix~$A$
of a hypergraph~$H$ with vertices~$v_1,\dots,v_n$
and hyperedges~$e_1,\dots,e_m$,
and the $(m\times m)$-matrix $\IE:=AA^T$.
Then,
\begin{align}
  \IE_{ij}&=\sum_{k=1}^nA_{ik}A_{jk}=|e_i\cap e_j|,
            \text{\qquad in particular,\qquad } \IE_{ii}=|e_i|.
            \label{mat-trick1}
\end{align}
Thus, if $\IE_{ij}=\IE_{ii}$,
then
we know $e_i\subseteq e_j$
and that \ref{SE} may be applicable.
More generally,
if $f(e_i)-(\IE_{ii}-\IE_{ij})\geq f(e_j)$,
we know that $e_i$ supersedes~$e_j$
and that \ref{DP} is applicable.
Similarly, for the $(n\times n)$-matrix $\IV:=A^TA$,
\begin{align}
  \IV_{ij}=\sum_{k=1}^mA_{ki}A_{kj}=|E(v_i)\cap E(v_j)|,
  \text{\qquad in particular,\qquad}
  \label{mat-trick2}
  \IV_{ii}=|E(v_i)|.
\end{align}
Thus,
if $\IV_{ij}=\IV_{jj}$,
then
$E(v_j)\subseteq E(v_i)$
and we know that~$v_i$ dominates~$v_j$ in the sense of \ref{MD}.
The trick for the efficient parallel algorithm is now that
matrix multiplication is efficiently parallelizable.
The trick for the sequential algorithm is
that the matrices~$\IE$ and~$\IV$,
once computed,
can be efficiently updated
on vertex and hyperedge deletion.

\subsection{Parallel algorithm}
\noindent
In this section,
we prove the following theorem.

\begin{theorem}\label{thm:parkern}
  Given an instance of Multiple Hitting Set
  $H=(V,E)$ with hyperedge demands
  $f:E\rightarrow\{1,\dots,\alpha\}$,
  \begin{enumerate}[(i)]
  \item\label{parkern1}   a problem kernel $H'=(V',E')$
    with $|V'|+|E'|\leq2\alpha\dn(\IG(H))$
    can be computed by an $\NCeins$ circuit family and
    
  \item\label{parkern2} $H$~can be exhaustively reduced with
    respect to both \ref{DP} and \ref{MD}
    in $O(\M(\IG(H)) \log |H|)$~time
    on $\poly(|H|)$ processors,
    where $\M(\IG(H))$ is the
    matching number of $\IG(H)$.
  \end{enumerate}
\end{theorem}

\noindent
\looseness=-1
The proof of \cref{thm:parkern} works as follows.
\cref{alg:par-edges}
exhaustively applies~\ref{DP},
thus realizes Step 1 of \cref{theorem:big-kernel}.
\cref{alg:par-vertices}
exhaustively applies~\ref{MD},
thus realizes Step 2 of \cref{theorem:big-kernel}.
Thus,
implementing \cref{alg:par-edges}
and \cref{alg:par-vertices}
as \NCeins{} circuit families,
we will prove \cref{thm:parkern}\eqref{parkern1}.
Then,
we will prove
\cref{thm:parkern}\eqref{parkern2}
using \cref{alg:par},
which applies \cref{alg:par-edges,alg:par-vertices}
until the hypergraph does not change.

\begin{algorithm}
  \label[algorithm]{alg:mat}
  \SetKwInOut{Input}{Input}
  \SetKwInOut{Output}{Result}
  \SetKwRepeat{Do}{do}{while}
  \SetKwFor{ForeachParallel}{foreach}{in parallel}
  
  \Input{Incidence matrix~$\mat A$ of a hypergraph $H = (V, E)$
    and hyperedge demands~$f\colon E\to\N$.}
  \Output{Incidence matrix of the hypergraph
    obtained via exhaustive application of
    \ref{DP} to~$H$.}

  $\IE \gets AA^T$.\tcp*{$\IE_{ij}=|e_i\cap e_j|$, as in \eqref{mat-trick1}}
  $D \gets (m\times m)$-matrix of all zeroes.\;
  $R \gets$ column vector of $m$ ones.\tcp*{$R_j=1\iff e_j$ will be in the output}
  \ForeachParallel(\tcp*[f]{$D_{ij}=1\iff e_i$ supersedes $e_j$}){$1\leq i,j\leq m$}{%
    \lIf(\nllabel{lin:dij}){$f(e_i) - (\IE_{ii} - \IE_{ij}) \geq f(e_j)$}{$D_{ij}\gets 1$.}}
  \ForeachParallel(\tcp*[f]{Exhaustive application of \ref{DP}}){$1\leq i,j\leq m$}{%
    \lIf(\nllabel{lin:dij1}){$(D_{ij}=1)\wedge((D_{ji}=0)\vee(i<j)))$}{$R_{j}\gets 0$.}
  }
  \Return{rows~$j$ of $A$ for which $R_j=1$.}
  \nllabel{lin:ri}

  \caption{Parallel algorithm for exhaustive application
    of \ref{DP}.}
  \label[algorithm]{alg:par-edges}
\end{algorithm}

The main challenge with the proof of \cref{thm:parkern} is that,
although the data reduction rules in \cref{sec:dr} are correct
when applied sequentially,
their independent parallel application may be wrong:
for example,
the algorithm may find that
a hyperedge~$e_j$ supersedes
a hyperedge~$e_i$ and vice versa,
and delete both.

\cref{alg:par-edges} applies \ref{DP}
using matrix multiplication and
\eqref{mat-trick1} %
to compute which hyperedges are superseded.
Herein,
the algorithm contains a tie breaker:
if the algorithm finds two hyperedges
superseding each other,
it deletes only the hyperedge with higher index.
We still have to show that
an application of \ref{DP} in this form
is correct and exhaustive.

\begin{lemma}\label[lemma]{alg1:correct}
  \cref{alg:par-edges} is correct.
\end{lemma}

\begin{proof}
  Let $H$~be the input hypergraph and~$H'$ the output hypergraph.
  We prove that
  \begin{enumerate}[(i)]
  \item\label{pec1} every multiple hitting set for~$H$ is one for~$H'$,
    
  \item\label{pec2} every multiple hitting set for~$H'$ is one for~$H$,
    
  \item\label{pec3} $H'$ is exhaustively reduced with respect to \ref{DP}.
  \end{enumerate}

  \noindent
  Claim~\eqref{pec1} is trivial since $H'$ is a sub-hypergraph of~$H$ and
  the algorithm does not change edge demands.
  For \eqref{pec2},
  we first prove that hyperedge supersedence is transitive.
  Let $e_i,e_j,e_k\in E$ be hyperedges.
  Observe that
  $e_k\setminus e_i \subseteq (e_k\setminus e_j)\cup(e_j\setminus e_i)$ and,
  thus,
  \begin{align}
  |e_k\setminus e_i| \leq |e_k\setminus e_j|+|e_j\setminus e_i|.
  \label{eq:intersections}
  \end{align}
  Assume that $e_i$ supersedes $e_j$ and that $e_j$ supersedes $e_k$,
  that is,
  $f(e_j) - |e_j\setminus e_i|\geq f(e_i)$ and
  $f(e_k) - |e_k\setminus e_j|\geq f(e_j)$.
  Adding up the two inequalities, we get
  $f(e_k) - (|e_k\setminus e_j|+|e_j\setminus e_i|)\geq f(e_i)$,
  and,
  using \eqref{eq:intersections},
  $f(e_k) - |e_k\setminus e_i|\geq f(e_i)$.
  Thus, $e_i$ supersedes $e_k$.
  In other words,
  in \cref{lin:dij} of \cref{alg:par-edges},
  we have that
  \begin{align}
    D_{ij}=1\text{ and }D_{jk}=1\text{ implies }D_{ik}=1.
    \label{eq:trans}
  \end{align}
  Now,
  consider the directed graph~$G$
  whose vertices are the hyperedges of~$H$
  and that contains an arc~$(e_i,e_j)$
  whenever, %
  in \cref{lin:dij1},
  \[
    D_{ij}=1\wedge((D_{ji}=0)\vee (i<j)).
  \]
  That is,
  if $e_i$ may cause the deletion of~$e_j$ in \cref{lin:dij1},
  then $(e_i,e_j)$ is an arc.
  Due to \eqref{eq:trans},
  this directed graph is acyclic.

  Finally,
  let $S$~be a multiple hitting set for~$H'$
  and let $e_j$~be a hyperedge in~$H$ that is not in~$H'$.
  Then $e_j$~has some source~$e_i$ as ancestor in~$G$.
  Nothing causes the deletion of~$e_i$, so it exists in~$H'$
  and, moreover,
  $e_i$ supersedes~$e_j$.
  Thus,
  since $S$~satisfies the demand of~$e_i$,
  it also satisfies the demand of~$e_j$.

  \eqref{pec3}
  Assume that $H'$~contains two hyperedges~$e_i$ and~$e_j$
  such that~$e_i$ supersedes~$e_j$.
  Then there is an arc between~$e_i$ and~$e_j$
  in the directed acyclic graph~$G$,
  contradicting the fact that
  $H'$~contains only hyperedges that are sources in~$G$. 
\end{proof}

\noindent
We have shown a parallel algorithm for exhaustively applying \ref{DP}.
To prove a parallel kernelization algorithm using \cref{theorem:big-kernel},
we still need to exhaustively apply \ref{MD},
for which we apply \cref{alg:par-vertices}.
It %
checks vertex dominance using matrix multiplication
and \eqref{mat-trick2}.
Again, a tie breaker is applied:
if it finds that a vertex~$v_i$ dominates a vertex~$v_j$
and vice versa,
then the vertex with lower index is considered
to dominate the vertex of higher index.

\begin{algorithm}
  \SetKwInOut{Input}{Input}
  \SetKwInOut{Output}{Result}
  \SetKwRepeat{Do}{do}{while}
  \SetKwFor{ForeachParallel}{foreach}{in parallel}
  
  \Input{Incidence matrix~$\mat A$ of a hypergraph $H = (V, E)$
    and hyperedge demands~$f\colon E\to\N$.}
  \Output{Incidence matrix of the hypergraph
    obtained from~$H$ via exhaustive application of
    \ref{MD}.}
  
  $\IV \gets A^TA.$\tcp*{$I_{ij}=|E(v_i)\cap E(v_j)|$, as by \eqref{mat-trick2}}
  $R\gets$ column vector of $n$~ones.\tcp*{$R_j=1\iff v_j$ will be in the output} 
  $D\gets (n\times n)$-matrix of all zeroes.\;
  \ForeachParallel(\tcp*[f]{$D_{ij}=1\implies v_i$ dominates~$v_j$}){$1\leq i,j\leq n$}{%
    \lIf(\nllabel{lin:dvij}){$(I_{ij}=I_{jj})\wedge((I_{ij}\ne I_{ii})\vee(i<j))$}{$D_{ij}\gets 1$.}
  }
    $C\gets (1,\dots,1)\cdot D.$\tcp*{$C_{j}=$ number of~$i$'s such that $D_{ij}=1$}\nllabel{lin:cjj}
  \ForeachParallel(\tcp*[f]{Delete $v_j$ if $C_{j}\geq \max_{e\in E(v_i)}f(e)$}){$1\leq j\leq n$}{%
    \lIf(\nllabel{lin:rvj}){$\bigwedge_{i=1}^m(A_{ij}=0\vee C_{j}\geq f(e_i))$}{$R_j\gets 0$.}
  }
  \Return{columns~$j$ of $A$ for which $R_j=1$.}
  \caption{Parallel algorithm for exhaustive application of \ref{MD}.}
  \label{alg:par-vertices}
\end{algorithm}

\begin{lemma}\label[lemma]{alg2:correct}
  \cref{alg:par-vertices} is correct.
\end{lemma}

\begin{proof}
    Let $H$~be the input hypergraph and~$H'$ the output hypergraph.
  We prove that
  \begin{enumerate}[(i)]
  \item\label{pvc1} for any multiple hitting set for~$H$,
    there is a multiple hitting set of at most the same size for~$H'$,
    
  \item\label{pvc2} every multiple hitting set for~$H'$ is one for~$H$,
    
  \item\label{pvc3} $H'$ is exhaustively reduced with respect to \ref{DP}.
  \end{enumerate}

  \noindent
  Claim \eqref{pvc2} follows
  since every vertex of~$H'$ is also in~$H$
  and both hypergraphs have the same hyperedges and demands.

  Towards \eqref{pvc1},
  consider a directed graph~$G=(V,A)$
  and having an arc~$(v_i,v_j)\in A$ whenever
  $D_{ij}=1$ after \cref{lin:dvij}.
  Since a vertex~$v_i$ dominates
  a vertex~$v_j$ if and only if~$I_{jj}=I_{ij}$,
  we get that
  $(v_i,v_j)\in A$ if and only if
  $v_i$ dominates~$v_j$ and,
  if $v_j$ also dominates~$v_i$,
  then $i<j$.
  The graph $G$~is acyclic:
  if there was a cycle~$L$,
  then any two vertices on~$L$ would dominate each other due to the transitivity
  of vertex dominance.
  Since the vertex index cannot always increase along the cycle,
  there is an edge $(v_i, v_j)$ on~$L$ such that $i>j$
  and the vertices dominate each other.
  This contradicts the rules by which we built~$G$.

  Now,
  assume that a multiple hitting set~$S$ for~$H$ contains
  some vertex~$v_j$ that is not in~$H'$.
  By \cref{obs:alive-ancestors},
  $v_j$~has some ancestor~$v_i$
  such that $v_i$~is not in~$H'$
  but all its ancestors in~$G$ are in~$H'$
  (possibly, $i=j$).
  Since any ancestor of~$v_i$ is also one of~$v_j$,
  and $v_i$ dominates~$v_j$,
  that is, $E(v_j)\subseteq E(v_i)$,
  it follows that there are at least
  \[
    \max_{e\in E(v_{i})}f(e)\geq \max_{e\in E(v_j)}f(e)
  \]
  ancestors of~$v_j$ left in~$H'$.
  We can replace~$v_j$ by one of them in~$S$ to get a
  multiple hitting set for~$H'$;
  if all of these ancestors are already in $S$,
  then we do not need~$v_j$ in the multiple hitting set at all.

  \eqref{pvc3} Assume that~$H'$ contains a vertex~$v_j$
  to which \ref{MD} is applicable.
  That is,
  $H'$~contains a set~$X$
  of at least $\max_{e\in E(v_j)}f(e)$ dominators of~$v_j$.
  If $D_{ij}=1$ for each~$v_i\in X$ in \cref{lin:dvij},
  then $C_{j}\geq|X|\geq\max_{e\in E(v_j)}f(e)$ in \cref{lin:cjj}
  and $v_j$~would have been deleted in \cref{lin:rvj}.
  Thus,
  there is some maximum~$i$
  such that $v_i\in X$ and $D_{ij}=0$.
  Since $v_i$ dominates~$v_j$ but $D_{ij}=0$,
  we know $D_{ji}=1$ and $j<i$.
  Thus,
  $v_j$ also dominates~$v_i$,
  that is,
  $E(v_j)=E(v_i)$.

  We now show that,
  in contradiction to the assumption that all vertices in~$X$
  are in~$H'$,
  \cref{alg:par-vertices}
  would have deleted~$v_i$ from~$H'$ in \cref{lin:rvj}.
  To this end,
  we show that
  $D_{ki}=1$
  for any~$v_k\in X'=(X\setminus\{v_i\})\cup\{v_j\})$ and, thus,
  in \cref{lin:cjj},
  \[
    C_{i}\geq|X'|=|X|\geq \max_{e\in E(v_j)}f(e)=\max_{e\in E(v_i)}f(e).
  \]
  For $k=j$,
  since $v_i$ dominates~$v_j$ and~$D_{ij}=0$,
  we have
  $D_{ki}=D_{ji}=1$.
  For $k\ne j$,
  $v_k$ dominates~$v_j$,
  since $v_k\in X$.
  Since $v_j$ dominates~$v_i$,
  we also know that $v_k$ dominates~$v_i$.
  If $D_{ki}=1$,
  then we are done.
  Otherwise,
  if $D_{ki}=0$,
  then we know $D_{ik}=1$, $j<i<k$,
  and $E(v_i)=E(v_j)=E(v_k)$.
  It follows that $v_k$ and $v_j$ dominate each other,
  that $j<k$, and thus $D_{kj}=0$ and $D_{jk}=1$.
  This contradicts the choice of~$i$.
\end{proof}

\noindent
We have proved the correctness of
\cref{alg:par-edges,alg:par-vertices}.
Together with a complexity analysis of these algorithms,
this will yield \cref{thm:parkern}\eqref{parkern1}.
We will apply them repeatedly as long as possible
to prove \cref{thm:parkern}\eqref{parkern2}.

\begin{algorithm}
  \SetKwInOut{Input}{Input}
  \SetKwInOut{Output}{Result}
  \SetKwRepeat{Do}{do}{while}
  \SetKwFor{ForeachParallel}{foreach}{in parallel}
  
  \Input{Incidence matrix~$\mat A$ of a hypergraph $H = (V, E)$
    and hyperedge demands~$f\colon E\to\N$.}
  \Output{Incidence matrix of the hypergraph obtained from~$H$
    via exhaustive application of \ref{DP} and \ref{MD}.}
  
  \Do{$A$ changes.}{
    $A'\gets$ apply \cref{alg:par-edges} to~$A$.\;
    $A\gets$ apply \cref{alg:par-vertices} to~$A'$.\;
  }
  \Return{$A$.}
  \caption{Algorithm for the proof of \cref{thm:parkern}\eqref{parkern2}.}
  \label{alg:par}
\end{algorithm}

\begin{proof}[Proof of \cref{thm:parkern}]
  \eqref{parkern1}
  By \cref{theorem:big-kernel,alg1:correct,alg2:correct},
  one application of \cref{alg:par-edges}
  followed by one application of \cref{alg:par-vertices}
  is enough to produce a problem kernel of the desired size.
  We argue that both algorithms can be realized
  by an $\NCeins$-circuit family.
  The key point is that
  integer multiplication, addition, subtraction, comparison,
  and integer matrix multiplication
  can be realized as \NCeins{}-circuits \citep{Coo85}.
  The parallel for loops
  can also be realized by \NCeins{}-circuits,
  using one subcircuit for each pair~$i,j$.
  Herein,
  the only thing noteworthy is that
  the ``$\wedge$''-operator
  in \cref{lin:rvj}
  of \cref{alg:par-vertices}
  can be realized by a tree
  of binary ``$\wedge$''-operators
  of depth $\log m$.

  \eqref{parkern2}
  It remains to analyze the number of iterations
  of \cref{alg:par}.
  Each iteration
  exhaustively applies first \ref{DP},
  then exhaustively applies \ref{MD}.
  If,
  during the $\ell$-th iteration,
  some hyperedge $e_j \in E$ is removed,
  then this is due to some hyperedge $e_i$
  superseding~$e_j$ at iteration~$\ell$
  but not at iteration~$\ell-1$.
  Thus,
  at iteration $\ell-1$, some vertex $v$
  contained in $e_i \setminus e_j$ was removed,
  so as to satisfy the condition
  $f(e_i) - |e_i\setminus e_j| \geq f(e_j)$
  at iteration~$\ell$.

  Consequently, for any iteration~$\ell$,
  except the first one,
  there is a pair $p_\ell = \{v^{(\ell)}, e^{(\ell)}\}$
  of a vertex~$v^{(\ell)}$ and a hyperedge~$e^{(\ell)} \in E$
  such that $v^{(\ell)} \in e^{(\ell)}$,
  both of which are deleted by the
  end of iteration~$\ell$.
  Observe that $p_\ell$~is an edge
  in the incidence graph~$\IG(H)$
  and that
  any two such edges~$p_i$ and~$p_j$ for $i<j$
  are disjoint.
  Consequently, the pairs~$p_i$
  form a matching in the incidence graph $\IG(H)$
  and the number of iterations cannot exceed~$\M(\IG(H)) + 1$.
\end{proof}

\subsection{Sequential algorithm}
\noindent
In this section,
we present a sequential algorithm
that exhaustively reduces hypergraphs
with respect to both \ref{DP} and \ref{MD}.
Its  running time is quadratic at worst
and matches the running time that one would expect
from first applying \ref{DP} exhaustively
and then applying \ref{MD} exhaustively,
which may be required to be repeated $\M(\IG(H))$ times
for an exhaustive application of both,
as we have seen in the proof of \cref{thm:parkern}.
Also note that the running time is subquadratic for sparse hypergraphs.

\begin{theorem}\label{thm:ser}
  \ref{DP} and \ref{MD}
  can be exhaustively applied
  in $O(|H|\cdot (|V|+|E|))$ time.
\end{theorem}

\noindent
To prove the theorem,
we first show
that \cref{alg:seq-edges,alg:seq-vertices} exhaustively apply \ref{DP} and \ref{MD}, respectively.
These algorithms are applied in a loop in \cref{alg:seq}.
We show that running time of \cref{alg:seq} satisfies the requirements of \cref{thm:ser},
which proves the theorem.

Each of the \cref{alg:seq-edges,alg:seq-vertices}
implements the exhaustive application of
one of the reduction rules \ref{DP} and \ref{MD}
in the same way as its parallel counterpart,
\cref{alg:par-edges} or \cref{alg:par-vertices}, respectively,
sequentalizing the parallel loops.
There are two important differences, however.
Firstly,
\cref{alg:seq-edges,alg:seq-vertices}
do not compute the matrices~$\IE$ and~$\IV$
from \eqref{mat-trick1} and \eqref{mat-trick2},
but get them as input
along with two Boolean arrays~$\RE$ and~$\RV$ such that
\begin{align}
  \RE_i=0&\iff e_i\in E \text{ is removed, and}
  \label{eq:removed-edges}\\
  \RV_i=0&\iff v_i\in V \text{ is removed.}
  \label{eq:removed-vertices}
\end{align}
It is the responsibility of \cref{alg:seq-edges,alg:seq-vertices}
to update $\IE$, $\IV$, $\RE$ and $\RV$ in place
after any hyperedge or vertex removal.

\looseness=-1
Secondly,
while the parallel algorithms
check all pairs of hyperedges
for supersedence and all pairs of vertices for dominance,
the sequential algorithms save time by narrowing their search.
When \cref{alg:seq-edges} removes a hyperedge~$e$,
it stores the indices of $e$'s vertices in a list~$\PV$.
Then,
\cref{alg:seq-vertices} searches dominators only for the vertices in~$\PV$.
Similarly,
\cref{alg:seq-vertices} stores indices of hyperedges from~$E(v)$
for any removed vertex $v$ in a list~$\PE$.
Then,
\cref{alg:seq-edges} only
searches for hyperedges that are superseded by those in~$\PE$.
More formally,
\cref{alg:seq-edges,alg:seq-vertices}
ensure the following invariant:
\begin{inv}
  \label[inv]{theinv}
  $\IE,\IV$ satisfy \eqref{mat-trick1} and \eqref{mat-trick2};
  $\RE,\RV$ satisfy \eqref{eq:removed-edges} and \eqref{eq:removed-vertices};
  $\PE$~contains the indices of a superset of hyperedges that supersede others in the sense of \ref{DP};
  $\PV$~contains the indices of a
  superset of vertices that can be removed by~\ref{MD}.
\end{inv}

\begin{algorithm}
    \SetKwInOut{Input}{Input}
    \SetKwInOut{Output}{Result}

    \Input{A demand function $f$; and $\IE,\IV,\RE,\RV,\PE,\PV$ satisfying \cref{theinv} for some hypergraph~$H$.}
    \Output{All inputs are updated in\hyp place
      so as to satisfy \cref{theinv}
      for the hypergraph obtained from~$H$ by exhaustive application of \ref{DP}.}

    $Q \gets$ empty list \tcp*{$Q$ stores indices of removed hyperedges}
    \For(\nllabel{alg:seq-edges:loop1}\tcp*[f]{iterate through all possibly superseded hyperedges}){$j\in\{1,\dots,m\}$ such that $\RE_j = 1$} {
        \For(\nllabel{alg:seq-edges:loop2}\tcp*[f]{iterate through all possibly superseding hyperedges}){$i\in \PE$ such that $i \neq j$ and $\RE_i = 1$} {
            $D_{ij} \gets$ \textbf{if} $f(e_i) - (\IE_{ii} - \IE_{ij}) \geq f(e_j)$ \textbf{then} 1 \textbf{else} 0\;
            $D_{ji} \gets$ \textbf{if} $f(e_j) - (\IE_{jj} - \IE_{ji}) \geq f(e_i)$ \textbf{then} 1 \textbf{else} 0\;
            \If(\nllabel{alg:seq-edges:apply-rule}\tcp*[f]{apply rule \ref{rule:pushing-demand} to $e_i$ and $e_j$}){$(D_{ij} = 1) \wedge ((D_{ji} = 0) \vee (i < j))$} {
                $Q\gets Q\cup\{j\}$\nllabel{alg:seq-edges:store-q}\;
                \Break
            }
        }
    }
    \For(\nllabel{alg:seq-edges:update}){$j\in Q$} {
        $\RE_j \gets 0$\nllabel{lin:req}\;
        \For(\nllabel{lin:ivupd}){$(v_i, v_k)\in e_j\times e_j$} {
            $\IV_{ik} \gets \IV_{ik} - 1$ \tcp*{update $I^V$ after $e_j$'s removal}
        }
        $\PV \gets \PV \cup \{i\in\{1,\dots,n\}\mid v_i\in e_j\}$\nllabel{alg:seq-edges:populate-pv} \tcp*{Removal of~$e_j$ may make \ref{MD}
          applicable to vertices in~$e_j$}
    }
    $\PE \gets$ empty list\nllabel{alg:seq-edges:clear-pe}\;

    \caption{Sequential algorithm for exhaustive application of \ref{DP}.}
    \label{alg:seq-edges}
\end{algorithm}

\begin{lemma}\label[lemma]{alg4:correct}
  \cref{alg:seq-edges} is correct.
\end{lemma}

\begin{proof}
  \looseness=-1
    We prove the algorithm correctness
    by showing that it removes exactly those hyperedges
    that \cref{alg:par-edges} removes.
    In order to simulate the parallel run of \cref{alg:par-edges},
    the sequential algorithm uses nested loops
    in \cref{alg:seq-edges:loop1,alg:seq-edges:loop2}
    applying the rule \ref{DP} in \cref{alg:seq-edges:apply-rule}
    exactly as in \cref{alg:par-edges}.
    Since \cref{theinv} guarantees
    that $\PE$~contains the indices of
    all hyperedges that can supersede others,
    these loops iterate through all pairs of hyperedges to check for supersedence.
    Therefore,
    in \cref{alg:seq-edges:store-q},
    the list~$Q$ contains the indices of
    all hyperedges that would be removed by \cref{alg:par-edges}.

    We now show that,
    after execution of the algorithm,
    the data structures
    satisfy \cref{theinv} with respect to the hypergraph
    obtained by removing the hyperedges in~$Q$.
    The update happens in \cref{alg:seq-edges:update}.
    First,
    each removed hyperedge $e_j$ is marked as removed in~$\RE$ in \cref{lin:req}.
    Then,
    in \cref{lin:ivupd},
    the matrix $\IV$ is updated using the fact that,
    by removing~$e_j$,
    we delete it from any intersection $E(v_i)\cap E(v_k)$
    for each $\{v_i, v_k\}\subseteq e_j$.
    Although, the matrix~$\IE$ should also be updated in the $j$-th row and the $j$-th column,
    these row and column will never be accessed since $\RE_j=0$,
    so we do not do this update.
    To see why the update of~$\PV$
    in \cref{alg:seq-edges:populate-pv}
    is correct,
    consider a vertex $v\in V$
    that becomes removable by \ref{MD} after hyperedges from~$Q$ were removed.
    Then $|\Dom(v)|\geq\max_{e \in E(v)}f(e)$ is true
    after the removal, but not before it.
    Thus,
    $E(v)$ must contain at least one edge from~$Q$.
    Therefore,
    we populate~$\PV$ with all vertices contained in hyperedges from~$Q$.
    Finally, $\PE$ is cleared in \cref{alg:seq-edges:clear-pe},
    since no hyperedge can supersede another
    after \ref{DP} is applied exhaustively.
\end{proof}

\begin{algorithm}
    \SetKwInOut{Input}{Input}
    \SetKwInOut{Output}{Result}

    \Input{a demand function $f$; $\IE,\IV,\RE,\RV,\PE,\PV$
      satisfying \cref{theinv}.}
    \Output{All inputs are updated in\hyp place
      so as to satisfy \cref{theinv}
      for the hypergraph obtained from~$H$ by exhaustive application of \ref{MD}.}

    $Q \gets$ empty list \tcp*{$Q$ stores indices of removed vertices}
    \For(\nllabel{alg:seq-vertices:loop1}\tcp*[f]{iterate through all potentially removable vertices}){$j \in \PV$ such that $\RV_j = 1$} {
        $c \gets \max_{e\in E(v_j)}f(e)$ \tcp*{how many dominators does $v_j$ need to apply \ref{rule:multiple-domination}}
        \For(\nllabel{alg:seq-vertices:loop2}\tcp*[f]{iterate through all potential dominators of $v_j$}){$i\in\{1,\dots,n\}$ such that $i\neq j$ and $\RV_i = 1$} {
            \If(\tcp*[f]{check for domination, with tie breaker}){$(\IV_{ij} = \IV_{jj})\wedge((\IV_{ij}\neq\IV_{ii})\vee(i<j))$} {
                $c \gets c - 1$\nllabel{alg:seq-vertices:count-doms}\;
                \If(\nllabel{alg:seq-vertices:apply-rule}\tcp*[f]{remove $v_j$ via \ref{rule:multiple-domination}}){$c = 0$} {
                    $Q\gets Q\cup\{j\}$\nllabel{alg:seq-vertices:store-q}\;
                    \Break
                }
            }
        }
    }

    \For(\nllabel{alg:seq-vertices:update}){$j\in Q$} {
        $\RV_j \gets 0$\;
        \For{$(e_i, e_k) \in E(v_j)\times E(v_j)$} {
            $\IE_{ik} \gets \IE_{ik} - 1$ \tcp*{update $I^E$ after $v_j$'s removal}
        }
        $\PE \gets \PE \cup \{i\in\{1,\dots,m\}\mid e_i\in E(v_j) \}$\nllabel{alg:seq-vertices:populate-pe} \tcp*{removal of $v_j$ may make \ref{DP}
          applicable to hyperedges in~$E(v_j)$}
    }
    $\PV \gets$ empty list\nllabel{alg:seq-vertices:clear-pv}\;

    \caption{Sequential algorithm for exhaustive application of \ref{MD}.}
    \label{alg:seq-vertices}
\end{algorithm}

\begin{lemma}\label[lemma]{alg5:correct}
  \cref{alg:seq-vertices} is correct.
\end{lemma}

\begin{proof}
  We prove the lemma by showing
  that the algorithm removes exactly the same vertices as
  \cref{alg:par-vertices}.
  To this end,
  the algorithm counts
  the dominators for a vertex in \cref{alg:seq-vertices:count-doms} and
  applies the rule \ref{MD} in \cref{alg:seq-vertices:apply-rule}
  exactly as in \cref{alg:par-vertices}.
  Since \cref{theinv} guarantees
  that $\PV$~contains the indices of all vertices
  that can be removed by \ref{MD},
  the loops iterate through all removable vertices
  and for each of them --- through all dominators.
  Therefore, all vertices that would be removed by \cref{alg:par-vertices}
  are stored in the list~$Q$ due to \cref{alg:seq-vertices:store-q}.

  \looseness=-1
  We now show that after the algorithm,
  all data structures satisfy \cref{theinv}
  with respect to the hypergraph obtained by removing the vertices in~$Q$.
  The data structures are updated in the loop in
  \cref{alg:seq-vertices:update}.
  First,
  for each $j\in Q$,
  vertex $v_j$ is marked as removed in~$\RV$.
  The matrix $\IE$ is updated using the fact that,
  by removing~$v_j$,
  we extract it from any intersection~$e_i\cap e_k$
  for each $\{e_i, e_k\}\subseteq E(v_k)$.
  Although the matrix~$\IV$
  should also be updated in the $j$-th row and the $j$-th column,
  this matrix entry will never be accessed since $\RV_j=0$,
  so we do not update it.
  To see why $\PE$~is updated correctly,
  consider two hyperedges $e_i,e_j\in E$,
  such that $e_i$ supersedes $e_j$ after removal of vertices in~$Q$
  but not before the removal.
  Then, $f(e_i)-|e_i\setminus e_j|\geq f(e_j)$ holds after removal,
  but not before removal.
  Thus, $|e_i\setminus e_j|$ has decreased,
  and since hyperedge sizes cannot increase,
  at least one vertex in~$e_i$ must have been removed.
  The index of this vertex is contained in~$Q$.
  Thus,
  we populate~$\PE$ with all hyperedges containing any vertex whose index is in~$Q$ in \cref{alg:seq-vertices:populate-pe}.
  Finally, $\PV$~is cleared in \cref{alg:seq-vertices:clear-pv},
  since no vertex can be removed by \ref{MD} after it is applied exhaustively.
\end{proof}

\begin{algorithm}
    \SetKwInOut{Input}{Input}
    \SetKwInOut{Output}{Output}
    \SetKwRepeat{Do}{do}{while}

    \Input{A hypergraph $H = (V, E)$ and a demand function $f$.}
    \Output{Hypergraph $H' = (V', E')$ obtained from~$H$
      by exhaustive application of \ref{DP} and \ref{MD}.}

    $\IE \gets(m\times m)$-matrix of zeroes\;
    \For(\nllabel{alg:seq:build-ie}){$k\in\{1,\dots,n\}$ and $(e_i, e_j)\in E(v_k)\times E(v_k)$} {
        $\IE_{ij} \gets \IE_{ij} + 1$ \tcp*{$\IE_{ij}=|e_i\cap e_j|$, as in \eqref{mat-trick1}}
    }
    $\IV \gets (n\times n)$-matrix of zeros\;
    \For(\nllabel{alg:seq:build-iv}){$k\in\{1,\dots,m\}$
      and $(v_i,v_j)\in e_k\times e_k$} {
        $\IV_{ij} \gets \IV_{ij} + 1$ \tcp*{$\IV_{ij}=|E(v_i)\cap E(v_j)|$, as in \eqref{mat-trick2}}
    }
    $\RE \gets$ length-$m$ array of ones \tcp*{$\RE_j=0\iff e_j$ is removed}
    $\RV \gets$ length-$m$ array of ones \tcp*{$\RV_j=0\iff v_j$ is removed}
    $\PE \gets [1,\dots,m]$\nllabel{alg:seq:init-pe}\tcp*{list of indices of potentially superseding hyperedges}
    $\PV \gets [1,\dots,n]$\nllabel{alg:seq:init-pv}\tcp*{list of indices of potentially removable vertices}
    \Do(\tcp*[f]{At this point, \cref{theinv} holds for the input hypergraph~$H$}\nllabel{lin:inv}){$\RE$ or $\RV$ changes} {
        Apply \cref{alg:seq-edges} to $(f, \IE, \IV, \RE, \RV, \PE, \PV)$\;
        Apply \cref{alg:seq-vertices} to $(f, \IE, \IV, \RE, \RV, \PE, \PV)$
    }
    $E' \gets \{e_j \in E\mid \RE_j = 1\}$\nllabel{lin:builde'}\;
    $V' \gets \{v_j \in V\mid \RV_j = 1\}$\nllabel{lin:buildv'}\;
    \Return $H'=(V', E')$
    \caption{Sequential algorithm for exhaustive application of \ref{DP} and \ref{MD}}
    \label{alg:seq}
\end{algorithm}

\begin{proof}[Proof of \cref{thm:ser}]
  \cref{alg:seq} applies \ref{DP} and \ref{MD} as long as possible.
  We first prove correctness of the algorithm.
  To this end,
  we first show that \cref{theinv} holds with respect to
  the input hypergraph~$H$ at the start of the loop in \cref{lin:inv}.
  The matrix~$\IE$ is computed in \cref{alg:seq:build-ie}
  using the following observation:
  each vertex $v\in V$ is counted exactly once
  in each intersection of hyperedges in $E(v)$.
  It thus satisfies \cref{theinv}.

  \looseness=-1
  Similarly, each hyperedge~$e\in E$ is counted exactly once
  in each intersection of vertices in~$e$,
  thus
  $\IV$~satisfies
  \cref{theinv}
  in \cref{alg:seq:build-iv}.
  In \cref{alg:seq:init-pe,alg:seq:init-pv}
  we put the indices of all hyperedges into~$\PE$
  and the indices of all vertices into~$\PV$,
  so that \cref{theinv} is trivially satisfied.
  \cref{alg4:correct,alg5:correct} ensure
  that \cref{theinv} holds after each iteration
  of the loop,
  in particular after the last iteration,
  after which none of the two algorithms
  reduces any more vertices or hyperedges.

  Thus,
  according to \cref{theinv},
  after the loop in \cref{lin:inv},
  in \cref{lin:builde',lin:buildv'},
  from $\RE$ and $\RV$ we extract
  the hyperedges and vertices of a hypergraph~$H'$
  that was obtained from~$H$ by exhaustive application
  of both \ref{DP} and \ref{MD}.
  
  We now analyze the running time of \cref{alg:seq}.
  Before \cref{lin:inv},
  the algorithm spends $O(|H|\cdot (|V|+|E|))$ time
  for establishing \cref{theinv},
  in particular for filling the matrices~$\IE$ and~$\IV$.
  Indeed, to fill matrix $\IE$
  the algorithm spends $O(|E(v_k)|^2)$ time for each vertex $v_k\in V$,
  which sums up to
  $O(\sum_{v_k\in V}|E(v_k)|^2)\in O(|E|\sum_{v_k\in V}|E(v_k)|)\in O(|E|\cdot|H|)$.
  Similarly, for filling $\IV$ the running time is
  $O(\sum_{e_k\in E}|e_k|^2)\in O(|V|\sum_{e_k\in E}|e_k|)\in O(|V|\cdot|H|)$.
  Next,
  we analyze the accumulated running time
  of \cref{alg:seq-edges,alg:seq-vertices}
  over all iterations.

  In \cref{alg:seq-edges}, the two nested loops
  in \cref{alg:seq-edges:loop1,alg:seq-edges:loop2}
  are executed in $O(|E|\cdot|\PE|)$ time.
  Any hyperedge $e_i$
  is placed into~$\PE$ once during the initialization
  of \cref{alg:seq}
  and when   \cref{alg:seq-vertices} removes one of its vertices,
  which happens at most once per vertex in~$e_i$.
  Thus, the accumulated running time of \cref{alg:seq-edges} is
  $O(|E|\sum_{e\in E}(1+|e_i|))\subseteq O(|E|\cdot|H|)$.
  The next loop in \cref{alg:seq-edges}
  is in \cref{alg:seq-edges:update}.
  It runs in time $O(|Q|+\sum_{i\in Q}|e_i|^2)$.
  Since each hyperedge is deleted (that is, in~$Q$) at most once,
  this is $O(|E|+\sum_{i\in Q}|e_i|^2)\subseteq O(|E|+|V|\sum_{i\in Q}|e_i|)\subseteq O(|E|\cdot|H|)$.
  
    It remains to analyze the time spent in \cref{alg:seq-vertices}
    during all iterations.
    \cref{alg:seq-vertices} has two nested loops
    in \cref{alg:seq-vertices:loop1,alg:seq-vertices:loop2},
    which run in $O(|V|\cdot|\PV|)$~time.
    A vertex~$v_i$ is placed in $\PV$
    once during the initialization of \cref{alg:seq}
    and when \cref{alg:seq-edges} removes a hyperedge 
    containing $v_i$.
    Thus, the amortized running time of \cref{alg:seq-vertices}
    is $O(|V|\sum_{v\in V}(1+|E(v_i)|))\subseteq O(|V|\cdot|H|)$.
    The loop in \cref{alg:seq-vertices:update}
    runs in $O(|Q|+\sum_{i\in Q}|E(v_i)|^2)$ time.
    Over all executions of \cref{alg:seq-vertices},
    this sums up to $O(|V|+\sum_{i\in Q}|E(v_i)^2|)\subseteq O(|V|+|E|\sum_{i\in Q}|E(v_i)|)\subseteq O(|V|\cdot|H|)$,
    since no vertex is removed twice.
    The overall running time of \cref{alg:seq} is therefore $O(|H|\cdot(|V|+|E|))$.
\end{proof}

\section{Experiments}
\label{sec:exp}
\noindent
In \cref{sec:dr},
we presented several data reduction rules for
Multiple Hitting Set.
Of these,
\cref{theorem:big-kernel} shows
that an exhaustive application of \ref{SE}
followed by an exhaustive application of \ref{MD}
is enough to yield a problem kernel.
It is easy to come up with examples
where the additional data reduction rules
have no additional data reduction effect,
for example,
in the case of unit demands.
Also,
we provided efficient implementations only of \ref{DP} (which supersedes \ref{SE}) and \ref{MD},
whereas application of the other data reduction rules,
namely \ref{FE} and \ref{LP},
is more time\hyp consuming.

Thus,
in this section,
we experimentally evaluate the effect of various combinations of the
data reduction rules presented in \cref{sec:dr}
and also analyze which combinations reach
real speed\hyp ups compared to the data reduction algorithms
built into general optimization tools such as CBC.

First,
in \cref{sec:exp-setup},
we describe our experimental setup.
Then,
in \cref{sec:rw},
we describe experiments on real-world data arising
in cancer drug design.
Finally,
in \cref{sec:generated},
we describe experiments on random hypergraphs modeling
transportation network optimization tasks.

\subsection{Experimental setup}
\label{sec:exp-setup}
\noindent
\looseness=-1
Experiments were run on an
AMD Ryzen 9 5950X 16-Core CPU,
GeForce RTX 3090 GPU,
on Ubuntu 18.04.6.
We implemented our algorithms in C++ and compiled the code
with GCC~9.2.1.
The parallel algorithm was implemented on the GPU using
OpenCL~2.2.\footnote{\url{https://www.khronos.org/opencl/}}
All implemented algorithms are wrapped into a Python package.\footnote{The source code
is freely available at \url{https://gitlab.com/PavelSmirnov/hs-dilworth-kernel}.}

\paragraph{ILP solving}
We measure not only the data reduction effect,
but also the effect that the data reduction has on the total
time of solving Multiple Hitting Set instances.
To this end,
after data reduction,
we solve Multiple Hitting Set instances
using the CBC ILP solver\footnote{\url{https://github.com/coin-or/Cbc}}
with the following ILP model:
\begin{align}
  \min\sum_{v\in V}x_v&\text{\quad s.\,t.}\label{LPm}\\
\notag{}    \sum\limits_{v\in e}x_v\geq f(e)&\text{\quad for all }e\in E\\
  \notag{} x_v\in\{0,1\}&\text{\quad for all }v\in V
\end{align}
Here, $x_v$ indicates whether $v$ is taken into the hitting set or not.

\paragraph{Implementation details}

The data reduction rule \ref{FE} is implemented
in a straightforward way.
To implement \ref{LP},
we compute the required lower bounds~$L_i$
for each edge~$e_i$ using CBC on an
LP relaxation of \eqref{LPm}.
The data reduction rules \ref{DP} and \ref{MD}
have been implemented in the following three variants.
\begin{itemize}[\algoGPUFE:]
\item[\algoGPU{}] is an implementation
  of the parallel algorithm \cref{alg:par} on the GPU
  using OpenCL.
  
\item[\algoCPU{}] is an implementation
  of the sequential \cref{alg:seq} on the CPU.
  
\item[\algoGPUFE{}] alternatingly applies \algoGPU{}
  and \ref{FE} until exhaustion.
\end{itemize}
The data reduction result of \algoGPU{} and \algoCPU{} is the same,
only their running times may differ. 
In the OpenCL implementation of \algoGPU{},
as in the description of 
\cref{alg:par-edges,alg:par-vertices},
each pair of hyperedges (or vertices)
is tested for supersedence (or dominance)
independently.
However,
supersedence (or dominance)
is checked by iterating over possibly common
vertices (or hyperedges) sequentially.
This does not pose a problem,
since the GPU does not have enough processors for full
parallelization anyway.

Before applying any of \cref{alg:par-edges,alg:par-vertices} in \cref{alg:par},
the incidence matrix of the hypergraph is updated
on the CPU,
so as to keep only vertices and hyperedges that are not yet deleted.
The $(m\times n)$-incidence
matrix is stored in a space\hyp efficient manner:
if $n\gg m$,
then
the column for each vertex consists of
$\lceil m/32\rceil$ integers
of 32 bit,
so that the bits in each integer indicate
the incidence relation of the vertex with 32 consecutive hyperedges.
This allows us not only to allocate the optimal
amount of memory for the incidence matrix,
but also apply efficient bitwise operations provided by the GPU.
This speeds up the linear run through the hyperedges
for each vertex pair,
which is the more time\hyp consuming step when $n\gg m$.
If $m\gg n$,
we would want to condense rows instead of columns.

\begin{figure}[p]
  \centering
  \ref{rulelegend}
  
    \begin{tikzpicture}[scale=1]
        \begin{axis}[xlabel={$\alpha$ (demand)},
                     ylabel={\% of remaining hyperedges (of 152)},
                     xlabel shift=-1pt,
                     ylabel shift=-1pt,
                     grid=both,
                     grid style={line width=.1pt, draw=gray!10},
                     major grid style={line width=.2pt,draw=gray!50},
                     minor tick num=5,
                     xmin=0,
                     xmax=60,
                     ymin=0,
                     ymax=80,
                     legend columns=-1,
                     legend style={/tikz/every even column/.append style={column sep=0.5cm}},
                     legend to name={rulelegend}]
            \addplot[thick,densely dotted] table[col sep=tab,x=demand,y expr=\thisrow{num_edges_worse}/\thisrow{num_edges_identity}*100]{cancer-gpu-presolver.csv};\addlegendentry[scale=0.75]{\ref{MD}+\ref{SE}}
            \addplot[dashed] table[col sep=tab,x=demand,y expr=\thisrow{num_edges_opencl}/\thisrow{num_edges_identity}*100]{cancer-gpu-presolver.csv};\addlegendentry[scale=0.75]{\ref{MD}+\ref{DP}}
            \addplot[dashdotted] table[col sep=tab,x=demand,y expr=\thisrow{num_edges_opencl_fe}/\thisrow{num_edges_identity}*100]{cancer-gpu-presolver.csv};\addlegendentry[scale=0.75]{\ref{MD}+\ref{DP}+\ref{FE}}
            \addplot[solid] table[col sep=tab,x=demand,y expr=\thisrow{num_edges_lb_fe}/\thisrow{num_edges_identity}*100]{cancer-gpu-presolver.csv};\addlegendentry[scale=0.75]{\ref{MD}+\ref{rul:lp}+\ref{FE}}
        \end{axis}
    \end{tikzpicture}\hfill
    \begin{tikzpicture}[scale=1]
        \begin{axis}[xlabel={$\alpha$ (demand)},
                     ylabel={\% of remaining vertices (of 7124)},
                     xlabel shift=-1pt,
                     ylabel shift=-1pt,
                     grid=both,
                     grid style={line width=.1pt, draw=gray!10},
                     major grid style={line width=.2pt,draw=gray!50},
                     minor tick num=5,
                     xmin=0,
                     xmax=60,
                     ymin=0,
                     ymax=15,
                     ylabel near ticks, yticklabel pos=right]
             \addplot[thick,densely dotted] table[col sep=tab,x=demand,y expr=\thisrow{num_vertices_worse}/\thisrow{num_vertices_identity}*100]{cancer-gpu-presolver.csv};
                                 \addplot[dashed] table[col sep=tab,x=demand,y expr=\thisrow{num_vertices_opencl}/\thisrow{num_vertices_identity}*100]{cancer-gpu-presolver.csv};
                                 \addplot[dashdotted] table[col sep=tab,x=demand,y expr=\thisrow{num_vertices_opencl_fe}/\thisrow{num_vertices_identity}*100]{cancer-gpu-presolver.csv};
                                 \addplot[solid] table[col sep=tab,x=demand,y expr=\thisrow{num_vertices_lb_fe}/\thisrow{num_vertices_identity}*100]{cancer-gpu-presolver.csv};
        \end{axis}
    \end{tikzpicture}
    \caption{Data reduction effect of various data reduction rules
      from \cref{sec:dr} on the real-world data set.
      The data reduction rules specified in the legend are applied exhaustively.
      On the right, the graph for \ref{MD}+\ref{DP} coincides with the graph for \ref{MD}+\ref{SE}.}
    \label{fig:cancer-reduction}
\end{figure}
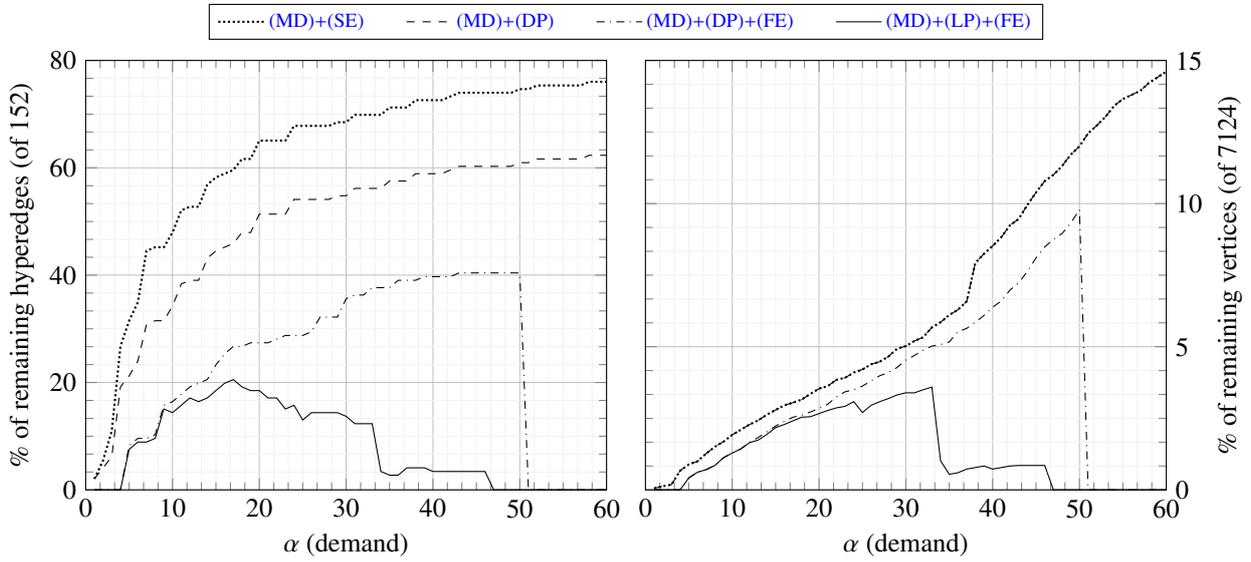

\begin{figure}[p]
  \centering
  \ref{algolegend}
  
    \begin{tikzpicture}[scale=0.95]
        \begin{axis}[xlabel={$\alpha$ (demand)},
                     ylabel={total time for reduction and CBC [s]},
                     xlabel shift=-1pt,
                     ylabel shift=-1pt,
                     grid=both,
                     grid style={line width=.1pt, draw=gray!10},
                     major grid style={line width=.2pt,draw=gray!50},
                     minor tick num=5,
                     xmin=0,
                     xmax=60,
                     ymin=0,
                     ymax=1,
                     log ticks with fixed point,
                     ymode=log,
                     scaled y ticks=false,
                     legend columns=-1,
                     legend style={/tikz/every even column/.append style={column sep=0.5cm}},
                     legend to name={algolegend}]
            \addplot[solid] table[col sep=tab,x=demand,y=identity_solve_time]{cancer-gpu-presolver.csv};\addlegendentry[scale=0.75]{no reduction}
            \addplot[thick,densely dotted] table[col sep=tab,x=demand,y=cpu_total_time]{cancer-gpu-presolver.csv};\addlegendentry[scale=0.75]{\algoCPU}
            \addplot[dashed] table[col sep=tab,x=demand,y=opencl_total_time]{cancer-gpu-presolver.csv};\addlegendentry[scale=0.75]{\algoGPU}
            \addplot[dashdotted] table[col sep=tab,x=demand,y=opencl_fe_total_time]{cancer-gpu-presolver.csv};\addlegendentry[scale=0.75]{\algoGPUFE}
        \end{axis}
    \end{tikzpicture}\hfill
    \begin{tikzpicture}[scale=0.95]
        \begin{axis}[xlabel={$\alpha$ (demand)},
                     ylabel={speedup factor},
                     xlabel shift=-1pt,
                     ylabel shift=-1pt,
                     grid=both,
                     grid style={line width=.1pt, draw=gray!10},
                     major grid style={line width=.2pt,draw=gray!50},
                     minor tick num=5,
                     xmin=0,
                     xmax=60,
                     ymin=0.8,
                     ymax=100,
                     log ticks with fixed point,
                     ymode=log,
                     ylabel near ticks, yticklabel pos=right]
                     \addplot[solid] coordinates{(0,1)(100,1)};
                     \addplot[thick,densely dotted] table[col sep=tab,x=demand,y=speed_up_factor_cpu]{cancer-gpu-presolver.csv};
                     \addplot[dashed] table[col sep=tab,x=demand,y=speed_up_factor_opencl]{cancer-gpu-presolver.csv};
                     \addplot[dashdotted] table[col sep=tab,x=demand,y=speed_up_factor_opencl_fe]{cancer-gpu-presolver.csv};
        \end{axis}
    \end{tikzpicture}
    \caption{Total running time for applying our data reduction algorithms to our Multiple Hitting Set instances from the real-world data set and solving them with CBC afterwards.}
    \label{fig:cancer-speedup}
\end{figure}
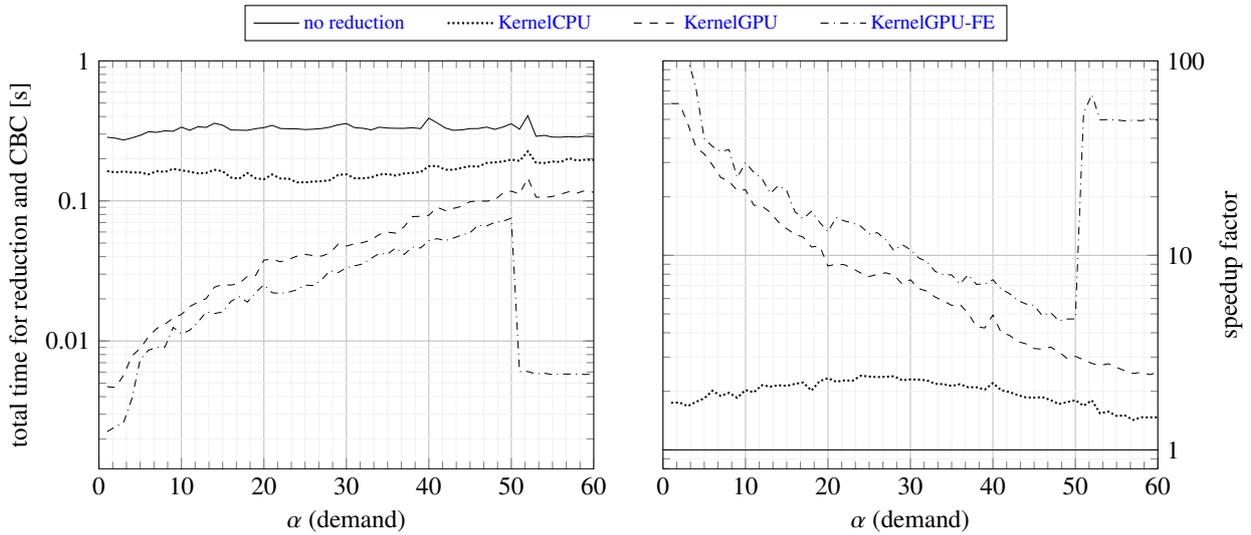

\subsection{Real-world data set}
\label{sec:rw}
\noindent
In this section,
we present experimental results on a
more recent version of the data set
used by \citet{Vaz09} and \citet{MPMM10}:
in these hypergraphs,
each hyperedge corresponds to a line of cancer cells,
each vertex corresponds to a chemical compound.
A vertex is contained in a hyperedge
if the corresponding chemical compound is observed
to inherit the growth of the corresponding cancer cell line.
Then,
Hitting Set solves the problem of selecting a minimum set
of compounds that inherit the growth of all cancer cell lines
in the data set.
\citet{MPMM10} motivate the use of \emph{Multiple} Hitting Set
by noise in the data:
to be on the safe side,
each cell line is hit not only once, but several times.

In more detail,
we downloaded the 2020 version
of the NCI60 human anti\hyp cancer drug screen set.\footnote{\url{http://dtp.nci.nih.gov/}}
This data set contains response data of
over 40,000 drugs against several cell lines.
In the same way as \citet{Vaz09} and \citet{MPMM10},
we add a chemical compound to a hyperedge,
corresponding to a cell line,
whenever its effect on the cell line
is stronger than the mean by more than two standard deviations.
As a result,
we obtain a hypergraph with 152 hyperedges
and 7124 vertices.

The hyperedge demands are governed by a parameter~$\alpha$
and are $f(e)=\min\{\alpha,|e|\}$ for each~$e\in E$.
We run experiments for all $\alpha\in\{1,\dots,60\}$,
since for larger $\alpha$ the instances
turned out to be trivial.

We chose this data set
since the resulting Multiple Hitting Set instances
are easily solved by CBC
and
we can thus see the trade\hyp offs between
power and speed of our data reduction rules.

\paragraph{Results}
\noindent
\cref{fig:cancer-reduction} shows the data reduction
effect of various combinations of the
data reduction rules in \cref{sec:dr}.
We see that \ref{DP} has a significantly stronger data reduction
effect than its weaker version \ref{SE}.
We also see that \ref{FE} has a strong additional data reduction
effect.
This is because it causes many cascading effects:
it can make
\ref{DP} and \ref{MD} applicable again.
In particular,
the combinations with \ref{FE} solve the problem
instance optimally for $\alpha\geq 51$ and for $\alpha\in\{1,2,3,4\}$,
whereas without \ref{FE}
the instances are solved optimally only for $\alpha\in\{1,2,3\}$.

\cref{fig:cancer-speedup}
shows how our data reduction algorithms
influence the speed for solving our Multiple Hitting Set instances
using CBC.
We see that the parallel implementations of \cref{alg:par}
on GPUs are able to speed up CBC by a factor of more than 10.
The speedup caused by \algoGPUFE{} is comparable to
that caused by \algoGPU{},
although it makes two to four (most commonly three) data reduction iterations.
For $\alpha\geq51$ the speed-up is exceptionally high
because \algoGPUFE{} solves the instance completely.
The sequential algorithm \algoCPU{}
speeds up CBC only by an observed factor of two.
The data reduction rule \ref{LP}
slows CBC down
on the considered data set. %
Thus,
in practise,
we can recommend \ref{LP} only for
hard Multiple Hitting Set instances,
since solving an LP relaxation
for each hyperedge in the input hypergraph
is rather expensive (but takes polynomial time).

The fact that \algoGPU{} speeds up CBC by an order of magnitude,
whereas \algoCPU{} gives yields speed-ups of only a factor of two,
shows that, in practice, parallelization in kernelization
can indeed make a significant difference.
In cases where the GPU
has too little memory to execute \algoGPU{},
it may still make sense to fall back to \algoCPU{}.

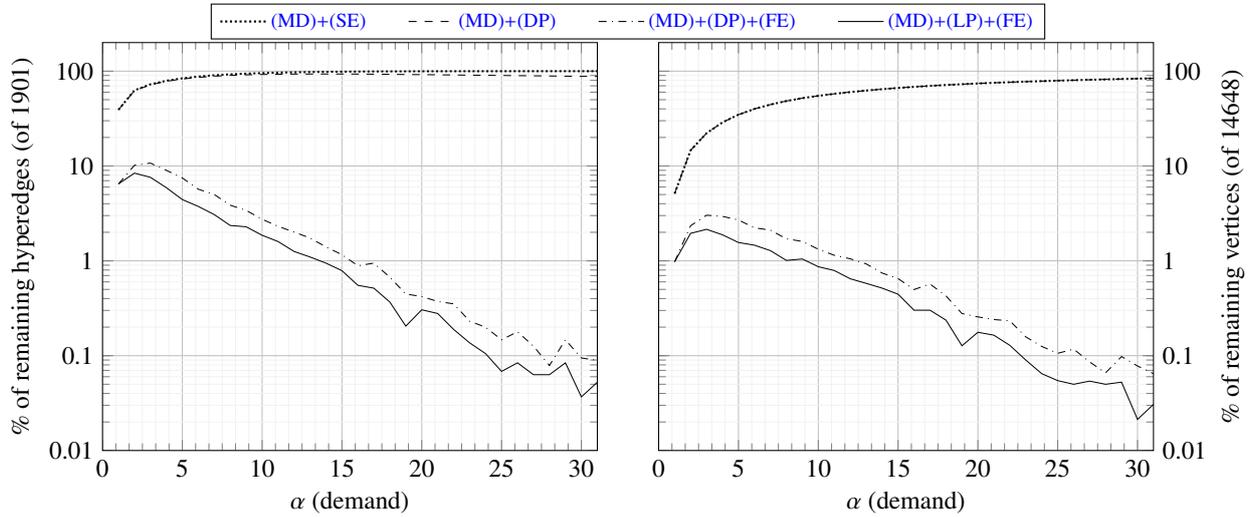
\begin{figure}[p]
  \centering
  \ref{rulelegend}
  
    \begin{tikzpicture}[scale=0.95]
        \begin{axis}[xlabel={$\alpha$ (demand)},
                     ylabel={\% of remaining hyperedges (of 1901)},
                     xlabel shift=-1pt,
                     ylabel shift=-1pt,
                     grid=both,
                     grid style={line width=.1pt, draw=gray!10},
                     major grid style={line width=.2pt,draw=gray!50},
                     minor tick num=5,
                     xmin=0,
                     xmax=31,
                     ymin=0.01,
                     ymax=200,
                     ymode=log,
                     log ticks with fixed point,
                     legend style={{at=(1, 0.8)},
                                   anchor=north east}]
            \addplot[thick,densely dotted] table[col sep=tab,x=demand,y expr=\thisrow{num_edges_worse}/\thisrow{num_edges_identity}*100]{big-girg.csv};%
            \addplot[dashed] table[col sep=tab,x=demand,y expr=\thisrow{num_edges_opencl}/\thisrow{num_edges_identity}*100]{big-girg.csv};%
            \addplot[dashdotted] table[col sep=tab,x=demand,y expr=\thisrow{num_edges_opencl_fe}/\thisrow{num_edges_identity}*100]{big-girg.csv};%
            \addplot[solid] table[col sep=tab,x=demand,y expr=\thisrow{num_edges_lb_fe}/\thisrow{num_edges_identity}*100]{big-girg.csv};%
        \end{axis}
    \end{tikzpicture}\hfill
    \begin{tikzpicture}[scale=0.95]
        \begin{axis}[xlabel={$\alpha$ (demand)},
                     ylabel={\% of remaining vertices (of 14648)},
                     xlabel shift=-1pt,
                     ylabel shift=-1pt,
                     grid=both,
                     grid style={line width=.1pt, draw=gray!10},
                     major grid style={line width=.2pt,draw=gray!50},
                     minor tick num=5,
                     xmin=0,
                     xmax=31,
                     ymin=0.01,
                     ymax=200,
                     ymode=log,
                     log ticks with fixed point,
                     ylabel near ticks, yticklabel pos=right,
                     legend style={{at=(1, 0.8)},
                                   anchor=north east}]
            \addplot[thick,densely dotted] table[col sep=tab,x=demand,y expr=\thisrow{num_vertices_worse}/\thisrow{num_vertices_identity}*100]{big-girg.csv};%
            \addplot[dashed] table[col sep=tab,x=demand,y expr=\thisrow{num_vertices_opencl}/\thisrow{num_vertices_identity}*100]{big-girg.csv};%
            \addplot[dashdotted] table[col sep=tab,x=demand,y expr=\thisrow{num_vertices_opencl_fe}/\thisrow{num_vertices_identity}*100]{big-girg.csv};%
            \addplot[solid] table[col sep=tab,x=demand,y expr=\thisrow{num_vertices_lb_fe}/\thisrow{num_vertices_identity}*100]{big-girg.csv};%
        \end{axis}
    \end{tikzpicture}
    \caption{Data reduction effect of various data reduction rules
      from \cref{sec:dr} on the generated data set.
      The data reduction rules specified in the legend are applied exhaustively;
      the graphs for \ref{MD}+\ref{SE} and \ref{MD}+\ref{DP} are nearly identical.
    }
    \label{fig:girg-reduction}
\end{figure}

\begin{figure}[p]
  \centering
  \ref{algolegend}
  
    \begin{tikzpicture}[scale=0.95]
        \begin{axis}[xlabel={$\alpha$ (demand)},
                     ylabel={total time for reduction and CBC [s]},
                     xlabel shift=-1pt,
                     ylabel shift=-1pt,
                     grid=both,
                     grid style={line width=.1pt, draw=gray!10},
                     major grid style={line width=.2pt,draw=gray!50},
                     minor tick num=5,
                     xmin=0,
                     xmax=31,
                     ymin=2,
                     ymax=10.1,
                     scaled y ticks=false]
                     \addplot[solid] table[col sep=tab,x=demand,y=identity_solve_time]{big-girg.csv};
                     \addplot[thick,densely dotted] table[col sep=tab,x=demand,y=cpu_total_time]{big-girg.csv};
                     \addplot[dashed] table[col sep=tab,x=demand,y=opencl_total_time]{big-girg.csv};
                     \addplot[dashdotted] table[col sep=tab,x=demand,y=opencl_fe_total_time]{big-girg.csv};
        \end{axis}
    \end{tikzpicture}\hfill
    \begin{tikzpicture}[scale=0.95]
        \begin{axis}[xlabel={$\alpha$ (demand)},
                     ylabel={speedup factor},
                     xlabel shift=-1pt,
                     ylabel shift=-1pt,
                     grid=both,
                     grid style={line width=.1pt, draw=gray!10},
                     major grid style={line width=.2pt,draw=gray!50},
                     minor tick num=5,
                     xmin=0,
                     xmax=31,
                     ymin=0.1,
                     ymax=10,
                     ymode=log,
                     log ticks with fixed point,
                     ylabel near ticks, yticklabel pos=right]
                     \addplot[solid] coordinates{(0,1)(100,1)};
            \addplot[thick,densely dotted] table[col sep=tab,x=demand,y=speed_up_factor_cpu]{big-girg.csv};
            \addplot[dashed] table[col sep=tab,x=demand,y=speed_up_factor_opencl]{big-girg.csv};
            \addplot[dashdotted] table[col sep=tab,x=demand,y=speed_up_factor_opencl_fe]{big-girg.csv};
        \end{axis}
    \end{tikzpicture}
    \caption{Total running time for applying our data reduction algorithms to our Multiple Hitting Set instances from the generated data set and solving them with CBC afterwards.}
    \label{fig:girg-speedup}
\end{figure}
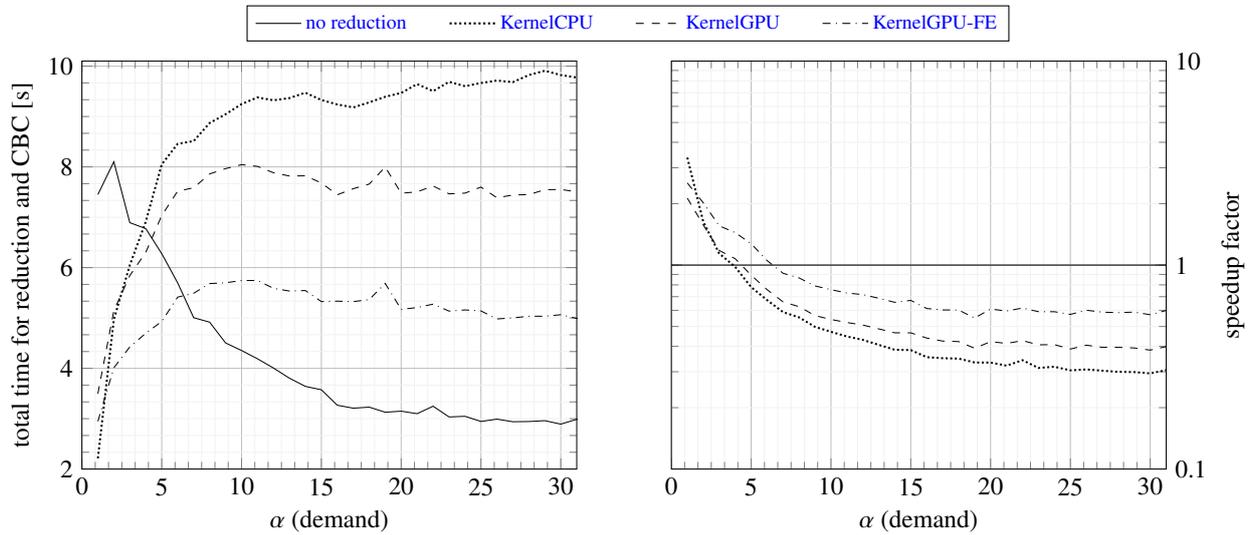

\subsection{Generated data set}
\label{sec:generated}
\noindent
In this section,
we present experimental results on
a data set generated
using a random hypergraph model proposed by \citet{BFFS19}.
According to their results,
choosing the following values for the model parameters
yields Hitting Set problem instances
close to those arising in real\hyp world public transportation
optimization problems:
\begin{itemize}[---]
\item $a=4\cdot10^{-5}$. This value affects the average vertex degree,
  which will be close to~$2$ on average.
\item $\beta=3.5$. This value controls the \emph{heterogeneity} of vertex degrees.
\item $T=0.5$. The \emph{temperature} influences the \emph{locality} of the network.
    In short,
    the geographic positions
    of a transportation network's vertices
    seem to be the cause of similarities
    in the network's connections (that is, hyperedges).
    For definitions and an in-depth analysis,
    we refer the reader to the original article of \citet{BFFS19}.
\end{itemize}
For each demand $\alpha\in\{1,\dots,30\}$,
we generated $10$~random instances with the above parameters
and an upper bound of $20\,000$~vertices and $2\,000$~hyperedges.
We removed possible empty hyperedges and isolated vertices from the generated hypergraphs
and thus obtained $10$~instances with $15\,014.8$~vertices and $1\,903.2$~hyperedges on average.
The average size of a hyperedge was about~$8$,
so we chose the maximum~$\alpha$ to be $30$.
On the obtained $10$~instances,
we ran the implementations described in \cref{sec:exp-setup}
for all values of~$\alpha$.
For each particular value of~$\alpha$,
we report the total numbers of vertices,
the total number of hyperedges,
the total data reduction time,
and the total CBC solution time,
over the $10$~instances.
The speed-up factor is reported with respect to the total solution time.

\paragraph{Results}
\cref{fig:girg-reduction} shows
the data reduction effect of various sets of data reduction rules.
This figure is similar to \cref{fig:cancer-reduction},
Compared to,
\cref{fig:cancer-reduction},
the gap between \ref{MD}+\ref{DP} and
\ref{MD}+\ref{DP}+\ref{FE} is much larger.
Whereas \ref{MD}+\ref{DP}
reduces the number of hyperedges to $40\,\%$ for $\alpha=1$,
and
is almost useless for $\alpha\geq 10$,
\ref{MD}+\ref{DP}+\ref{FE}
reduces the number of hyperedges to at most $10\,\%$
and leaves at most $1\,\%$ of the hyperedges for larger~$\alpha$.
A similar behaviour can be seen for the vertex reduction.

Concerning the running time,
\cref{fig:girg-speedup} shows
that, generally speaking,
among our data reduction algorithms,
the \algoGPUFE{} variant yields the highest speed-up,
again showing the feasibility of implementing kernelization on the GPU.
However,
our data reduction algorithms
yield speed-ups only for small values of~$\alpha$.
For $\alpha\geq7$,
it is better to apply CBC without data reduction.
Regarding rule \ref{LP},
it did not yield any speed-up
on this data set either
and we hence omitted it from the plot.

\section{Conclusion}
\noindent
\looseness=-1
We contributed to \citeauthor{Wei98}'s
data reduction algorithm for Hitting Set in three ways:
we proved that it yields a problem kernel
with a number of vertices and hyperedges
linear in the Dilworth number of the incidence graph;
we presented efficient parallel and sequential realizations
and experimentally evaluated them;
and
we generalized it to Multiple Hitting Set.
However,
the problem kernel with respect to the Dilworth number is just
a first step to understanding the structure of hypergraphs
that are well reducible by Weihe's algorithm:
the problem kernel is obtained
already after one exhaustive application
of \ref{DP} and one exhaustive application of \ref{MD},
yet the data reduction rules can be applied repeatedly
to shrink the hypergraph,
even more so in combination with \ref{FE}.
This naturally raises two questions:
Which natural hypergraph
parameter bounds the size of the reduced instance
after exhaustive applications of
all of \ref{DP}, \ref{MD}, and \ref{FE}?
Can the exhaustive application of all three of them be realized
effectively, say, on \NCeins{} circuits,
or in quadratic sequential time?

Besides this,
we can draw three more general conclusions from our work.
The first is that it seems worthwhile to study
the parameterized complexity of problems
with respect to the Dilworth number,
rather than with respect to the neighborhood diversity,
which is already well\hyp studied.

The second is that,
both, theoretical bounds on problem kernel sizes,
as well as empirical kernel size measurements,
are insufficient to make conclusions
on the effect of data reduction on the running time of solving
real problem instances.
In the case where problem instances are already solved quite
well,
we have seen that quite some additional effort
(in form of parallelization on GPUs)
might be required to speed up the solution process.

\looseness=-1
The third is what made it feasible
to implement our kernelization algorithms on GPUs:
note that we have shown a parallel kernelization algorithm
in terms of \NCeins{} circuits,
where each gate executes its own instruction on its own data,
whereas GPUs operate
in terms of the SIMD (Single Instruction, Multiple Data) model:
all cores of one multiprocessor perform the same operations,
yet on different data.
Data reduction algorithms seem to lend themselves
well to realization on GPUs
since, often,
they consist of a fixed set of data reduction rules,
applied to different parts of the data.
We thus expect that,
although previously more often studied
in the context of attacking NP\hyp hard problems \citep{BST15,BST20,BT20},
parallel kernelization will have a stronger real\hyp world impact
in the context of speeding up polynomial\hyp time algorithms
on large data sets,
such as linear\hyp time data reduction
was applied to speed up matching algorithms \citep{MNN20},
or in the context of designing
parallel fixed\hyp parameter algorithms \citep{BST15} for
P-complete problems,
which do not give in to massive parallelization.

\paragraph{Funding}
This work was initiated during
the First Workshop of Mathematical Center in Akademgorodok,
which was held during July 13--17 and August 10--14, 2020,
under support of the Ministry of Science and
Higher Education of the Russian Federation,
agreement No.\ 075-15-2019-167.
Until the end of 2020,
R.\ van Bevern,
P.\ V.\ Smirnov, and O.\ Yu.\ Tsidulko
were supported by Russian Foundation for Basic Research,
project 18-501-12031 NNIO\textunderscore a.

\bibliographystyle{hs-workshop}
\bibliography{hs-workshop}

\begin{thebibliography}{53}
\providecommand{\natexlab}[1]{#1}
\providecommand{\url}[1]{\texttt{#1}}
\providecommand{\urlprefix}{URL }
\expandafter\ifx\csname urlstyle\endcsname\relax
  \providecommand{\doi}[1]{doi:\discretionary{}{}{}#1}\else
  \providecommand{\doi}{doi:\discretionary{}{}{}\begingroup
  \urlstyle{rm}\Url}\fi
\providecommand{\selectlanguage}[1]{\relax}
\providecommand{\eprint}[2][]{\url{#2}}

\bibitem[{Abu-Khzam(2010)}]{Abu10}
F.~N. Abu-Khzam, \emph{A kernelization algorithm for {$d$-Hitting Set}},
  Journal of Computer and System Sciences \textbf{76} (2010), 524--531,
  \doi{10.1016/j.jcss.2009.09.002}.

\bibitem[{Abu-Khzam et~al.(2022)Abu-Khzam, Lamm, Mnich, Noe, Schulz, and
  Strash}]{ALM+22}
F.~N. Abu-Khzam, S.~Lamm, M.~Mnich, A.~Noe, C.~Schulz, and D.~Strash,
  \emph{Recent advances in practical data reduction}, H.~Bast, C.~Korzen,
  U.~Meyer, and M.~Penschuck (eds.), \emph{Algorithms for Big Data}, Springer,
  \emph{LNCS}, vol. 13201, 2022, pp. 97--133,
  \doi{10.1007/978-3-031-21534-6_6}.

\bibitem[{Achterberg et~al.(2020)Achterberg, Bixby, Gu, Rothberg, and
  Weninger}]{ABG+20}
T.~Achterberg, R.~E. Bixby, Z.~Gu, E.~Rothberg, and D.~Weninger, \emph{Presolve
  reductions in mixed integer programming}, {INFORMS} Journal on Computing
  \textbf{32} (2020), 473--506, \doi{10.1287/ijoc.2018.0857}.

\bibitem[{Agrawal et~al.(2020)Agrawal, Aravind, Kalyanasundaram, Kare, Lauri,
  Misra, and Reddy}]{AAK+20}
A.~Agrawal, N.~Aravind, S.~Kalyanasundaram, A.~S. Kare, J.~Lauri, N.~Misra, and
  I.~V. Reddy, \emph{Parameterized complexity of happy coloring problems},
  Theoretical Computer Science  (2020), \doi{10.1016/j.tcs.2020.06.002}.

\bibitem[{Alman et~al.(2020)Alman, Mnich, and Vassilevska~Williams}]{AMW20}
J.~Alman, M.~Mnich, and V.~Vassilevska~Williams, \emph{Dynamic parameterized
  problems and algorithms}, ACM Transactions on Algorithms \textbf{16} (2020),
  \doi{10.1145/3395037}.

\bibitem[{Arora and Barak(2009)}]{AB09}
S.~Arora and B.~Barak, \emph{Computational Complexity: A modern approach},
  Cambridge University Press, 2009.

\bibitem[{Bannach et~al.(2022)Bannach, Heinrich, Reischuk, and Tantau}]{BHRT22}
M.~Bannach, Z.~Heinrich, R.~Reischuk, and T.~Tantau, \emph{Dynamic kernels for
  hitting sets and set packing}, Algorithmica \textbf{84} (2022), 3459--3488,
  \doi{10.1007/s00453-022-00986-0}.

\bibitem[{Bannach et~al.(2020)Bannach, Skambath, and Tantau}]{BST20}
M.~Bannach, M.~Skambath, and T.~Tantau, \emph{Kernelizing the hitting set
  problem in linear sequential and constant parallel time}, S.~Albers (ed.),
  \emph{SWAT 2020}, Schloss Dagstuhl--Leibniz-Zentrum f{\"u}r Informatik,
  \emph{LIPIcs}, vol. 162, 2020, pp. 9:1--9:16,
  \doi{10.4230/LIPIcs.SWAT.2020.9}.

\bibitem[{Bannach et~al.(2015)Bannach, Stockhusen, and Tantau}]{BST15}
M.~Bannach, C.~Stockhusen, and T.~Tantau, \emph{Fast parallel fixed-parameter
  algorithms via color coding}, T.~Husfeldt and I.~Kanj (eds.), \emph{IPEC
  2015}, Schloss Dagstuhl--Leibniz-Zentrum f{\"u}r Informatik, 2015,
  \emph{LIPIcs}, vol.~43, \doi{10.4230/LIPICS.IPEC.2015.224}.

\bibitem[{Bannach and Tantau(2020)}]{BT20}
M.~Bannach and T.~Tantau, \emph{Computing hitting set kernels by
  {AC$_0$}-circuits}, Theory of Computing Systems \textbf{62} (2020), 374--399,
  \doi{10.1007/s00224-019-09941-z}.

\bibitem[{van Bevern et~al.(2020)van Bevern, Fluschnik, and Tsidulko}]{BFT20b}
R.~van Bevern, T.~Fluschnik, and O.~{\relax Yu}. Tsidulko, \emph{On approximate
  data reduction for the {Rural Postman Problem}: Theory and experiments},
  Networks \textbf{76} (2020), 485--508, \doi{10.1002/net.21985}.

\bibitem[{van Bevern et~al.(in press)van Bevern, Melnikov, Smirnov, and
  Tsidulko}]{BMST23}
R.~van Bevern, A.~Melnikov, P.~Smirnov, and O.~Tsidulko, \emph{On data
  reduction for dynamic vector bin packing}, Operations Research Letters  (in
  press), \doi{10.1016/j.orl.2023.06.005}.

\bibitem[{van Bevern et~al.(2012)van Bevern, Moser, and Niedermeier}]{BMN12}
R.~van Bevern, H.~Moser, and R.~Niedermeier, \emph{Approximation and
  tidying---a problem kernel for {$s$-Plex Cluster Vertex Deletion}},
  Algorithmica \textbf{62} (2012), 930--950, \doi{10.1007/s00453-011-9492-7}.

\bibitem[{van Bevern and Smirnov(2020)}]{BS20}
R.~van Bevern and P.~V. Smirnov, \emph{Optimal-size problem kernels for
  $d$-hitting set in linear time and space}, Information Processing Letters
  \textbf{163} (2020), 105\,998, \doi{10.1016/j.ipl.2020.105998}.

\bibitem[{Bläsius et~al.(2019{\natexlab{a}})Bläsius, Fischbeck, Friedrich,
  and Schirneck}]{BFFS19}
T.~Bläsius, P.~Fischbeck, T.~Friedrich, and M.~Schirneck, \emph{Understanding
  the effectiveness of data reduction in public transportation networks},
  K.~Avrachenkov, P.~Prałat, and N.~Ye (eds.), \emph{WAW 2019}, Springer,
  \emph{LNCS}, vol. 11631, 2019{\natexlab{a}}, pp. 87--101,
  \doi{10.1007/978-3-030-25070-6_7}.

\bibitem[{Bläsius et~al.(2019{\natexlab{b}})Bläsius, Friedrich, Lischeid,
  Meeks, and Schirneck}]{BFL+19}
T.~Bläsius, T.~Friedrich, J.~Lischeid, K.~Meeks, and M.~Schirneck,
  \emph{Efficiently enumerating hitting sets of hypergraphs arising in data
  profiling}, S.~Kobourov and H.~Meyerhenke (eds.), \emph{ALENEX 2019}, SIAM,
  2019{\natexlab{b}}, pp. 130--143, \doi{10.1137/1.9781611975499.11}.

\bibitem[{Brewka et~al.(2019)Brewka, Thimm, and Ulbricht}]{BTU19}
G.~Brewka, M.~Thimm, and M.~Ulbricht, \emph{Strong inconsistency}, Artificial
  Intelligence \textbf{267} (2019), 78--117,
  \doi{10.1016/j.artint.2018.11.002}.

\bibitem[{Calamoneri and Petreschi(2014)}]{CP14}
T.~Calamoneri and R.~Petreschi, \emph{On pairwise compatibility graphs having
  {Dilworth} number $k$}, Theoretical Computer Science \textbf{547} (2014),
  82--89, \doi{10.1016/j.tcs.2014.06.024}.

\bibitem[{Cook(1985)}]{Coo85}
S.~A. Cook, \emph{A taxonomy of problems with fast parallel algorithms},
  Information and Control \textbf{64} (1985), 2--22,
  \doi{10.1016/S0019-9958(85)80041-3}.

\bibitem[{Cotta et~al.(2004)Cotta, Sloper, and Moscato}]{CSM04}
C.~Cotta, C.~Sloper, and P.~Moscato, \emph{Evolutionary search of thresholds
  for robust feature set selection: Application to the analysis of microarray
  data}, G.~R. Raidl, S.~Cagnoni, J.~Branke, D.~W. Corne, R.~Drechsler, Y.~Jin,
  C.~G. Johnson, P.~Machado, E.~Marchiori, F.~Rothlauf, G.~D. Smith, and
  G.~Squillero (eds.), \emph{Applications of Evolutionary Computing}, Springer,
  2004, pp. 21--30, \doi{10.1007/978-3-540-24653-4_3}.

\bibitem[{Damaschke(2006)}]{Dam06}
P.~Damaschke, \emph{Parameterized enumeration, transversals, and imperfect
  phylogeny reconstruction}, Theoretical Computer Science \textbf{351} (2006),
  337--350, \doi{10.1016/j.tcs.2005.10.004}.

\bibitem[{Dell and van Melkebeek(2014)}]{DM14}
H.~Dell and D.~van Melkebeek, \emph{Satisfiability allows no nontrivial
  sparsification unless the polynomial-time hierarchy collapses}, Journal of
  the ACM \textbf{61} (2014), 23:1--23:27, \doi{10.1145/2629620}.

\bibitem[{Dilworth(1950)}]{Dil50}
R.~P. Dilworth, \emph{A decomposition theorem for partially ordered sets}, The
  Annals of Mathematics \textbf{51} (1950), 161, \doi{10.2307/1969503}.

\bibitem[{Dom et~al.(2014)Dom, Lokshtanov, and Saurabh}]{DLS14}
M.~Dom, D.~Lokshtanov, and S.~Saurabh, \emph{Kernelization lower bounds through
  colors and {IDs}}, {ACM} Transactions on Algorithms \textbf{11} (2014), 13,
  \doi{10.1145/2650261}.

\bibitem[{Fafianie and Kratsch(2014)}]{FK14}
S.~Fafianie and S.~Kratsch, \emph{Streaming kernelization}, E.~Csuhaj-Varjú,
  M.~Dietzfelbinger, and Z.~Ésik (eds.), \emph{MFCS 2014}, Springer,
  \emph{LNCS}, vol. 8635, 2014, pp. 275--286,
  \doi{10.1007/978-3-662-44465-8_24}.

\bibitem[{Fazekas et~al.(2018)Fazekas, Bacchus, and Biere}]{FBB18}
K.~Fazekas, F.~Bacchus, and A.~Biere, \emph{Implicit hitting set algorithms for
  maximum satisfiability modulo theories}, D.~Galmiche, S.~Schulz, and
  R.~Sebastiani (eds.), \emph{IJCAR 2018}, Springer, \emph{LNCS}, vol. 10900,
  2018, pp. 134--151, \doi{10.1007/978-3-319-94205-6_10}.

\bibitem[{Fellows(2002)}]{Fel02}
M.~R. Fellows, \emph{Parameterized complexity: The main ideas and connections
  to practical computing}, R.~Fleischer, B.~Moret, and E.~M. Schmidt (eds.),
  \emph{Experimental Algorithmics: From Algorithm Design to Robust and
  Efficient Software}, Springer, 2002, pp. 51--77,
  \doi{10.1007/3-540-36383-1_3}.

\bibitem[{Felsner et~al.(2003)Felsner, Raghavan, and Spinrad}]{FRS03}
S.~Felsner, V.~Raghavan, and J.~Spinrad, \emph{Recognition algorithms for
  orders of small width and graphs of small {Dilworth} number}, Order
  \textbf{20} (2003), 351--364, \doi{10.1023/B:ORDE.0000034609.99940.fb}.

\bibitem[{Flum and Grohe(2006)}]{FG06}
J.~Flum and M.~Grohe, \emph{Parameterized Complexity Theory}, Texts in
  Theoretical Computer Science, An EATCS Series, Springer, 2006,
  \doi{10.1007/3-540-29953-X}.

\bibitem[{Foldes and Hammer(1978)}]{FH78}
S.~Foldes and P.~L. Hammer, \emph{The {Dilworth} number of a graph},
  B.~Alspach, P.~Hell, and D.~Miller (eds.), \emph{Algorithmic Aspects of
  Combinatorics}, Elsevier, \emph{Annals of Discrete Mathematics}, vol.~2,
  1978, pp. 211--219, \doi{https://doi.org/10.1016/S0167-5060(08)70334-0}.

\bibitem[{Froese et~al.(2016)Froese, van Bevern, Niedermeier, and
  Sorge}]{FBNS16}
V.~Froese, R.~van Bevern, R.~Niedermeier, and M.~Sorge, \emph{Exploiting hidden
  structure in selecting dimensions that distinguish vectors}, Journal of
  Computer and System Sciences \textbf{82} (2016), 521--535,
  \doi{10.1016/j.jcss.2015.11.011}.

\bibitem[{Ganian(2012)}]{Gan12}
R.~Ganian, \emph{Using neighborhood diversity to solve hard problems},
  Available on arXiv, 2012, \urlprefix\url{http://arxig.org/abs/1201.3091}.

\bibitem[{Gargano and Rescigno(2015)}]{GR15}
L.~Gargano and A.~A. Rescigno, \emph{Complexity of conflict-free colorings of
  graphs}, Theoretical Computer Science \textbf{566} (2015), 39--49,
  \doi{10.1016/j.tcs.2014.11.029}.

\bibitem[{Gavenčiak et~al.(2020)Gavenčiak, Koutecký, and Knop}]{GKK20}
T.~Gavenčiak, M.~Koutecký, and D.~Knop, \emph{Integer programming in
  parameterized complexity: Five miniatures}, Discrete Optimization  (2020),
  100\,596, \doi{10.1016/j.disopt.2020.100596}.

\bibitem[{Gera et~al.(2018)Gera, Haynes, Hedetniemi, and Henning}]{GHHH18}
R.~Gera, T.~W. Haynes, S.~T. Hedetniemi, and M.~A. Henning, \emph{An annotated
  glossary of graph theory parameters, with conjectures}, \emph{Graph Theory},
  Springer, 2018, pp. 177--281, \doi{10.1007/978-3-319-97686-0_14}.

\bibitem[{Hoáng and Mahadev(1989)}]{HM89}
C.~Hoáng and N.~Mahadev, \emph{A note on perfect orders}, Discrete Mathematics
  \textbf{74} (1989), 77--84, \doi{10.1016/0012-365X(89)90200-8}.

\bibitem[{Hüffner et~al.(2010)Hüffner, Komusiewicz, Moser, and
  Niedermeier}]{HKMN10}
F.~Hüffner, C.~Komusiewicz, H.~Moser, and R.~Niedermeier,
  \emph{Fixed-parameter algorithms for cluster vertex deletion}, Theory of
  Computing Systems \textbf{47} (2010), 196--217,
  \doi{10.1007/s00224-008-9150-x}.

\bibitem[{Karp(1972)}]{Kar72}
R.~M. Karp, \emph{Reducibility among combinatorial problems}, R.~E. Miller,
  J.~W. Thatcher, and J.~D. Bohlinger (eds.), \emph{Complexity of Computer
  Computations}, The {IBM} Research Symposia Series, Springer, 1972, pp.
  85--103, \doi{10.1007/978-1-4684-2001-2_9}.

\bibitem[{Kratsch(2012)}]{Kra12}
S.~Kratsch, \emph{Polynomial kernelizations for {MIN F$^+\Pi_1$} and {MAX~NP}},
  Algorithmica \textbf{63} (2012), 532--550, \doi{10.1007/s00453-011-9559-5}.

\bibitem[{Lampis(2012)}]{Lam12}
M.~Lampis, \emph{Algorithmic meta-theorems for restrictions of treewidth},
  Algorithmica  (2012), 19--37, \doi{10.1007/s00453-011-9554-x}.

\bibitem[{Mathieson et~al.(2017)Mathieson, Mendes, Marsden, Pond, and
  Moscato}]{MMM+16}
L.~Mathieson, A.~Mendes, J.~Marsden, J.~Pond, and P.~Moscato,
  \emph{Computer-aided breast cancer diagnosis with optimal feature sets:
  Reduction rules and optimization techniques}, J.~M. Keith (ed.),
  \emph{Bioinformatics: Volume II: Structure, Function, and Applications},
  Springer, New York, NY, 2017, pp. 299--325,
  \doi{10.1007/978-1-4939-6613-4_17}.

\bibitem[{Mellor et~al.(2010)Mellor, Prieto, Mathieson, and Moscato}]{MPMM10}
D.~Mellor, E.~Prieto, L.~Mathieson, and P.~Moscato, \emph{A kernelisation
  approach for multiple $d$-hitting set and its application in optimal
  multi-drug therapeutic combinations}, PLOS ONE \textbf{5} (2010),
  \doi{10.1371/journal.pone.0013055}.

\bibitem[{Mertzios et~al.(2020)Mertzios, Nichterlein, and Niedermeier}]{MNN20}
G.~B. Mertzios, A.~Nichterlein, and R.~Niedermeier, \emph{The power of
  linear-time data reduction for maximum matching}, Algorithmica \textbf{82}
  (2020), 3521--3565, \doi{10.1007/s00453-020-00736-0}.

\bibitem[{Moreno-Centeno and Karp(2013)}]{MK13}
E.~Moreno-Centeno and R.~M. Karp, \emph{The implicit hitting set approach to
  solve combinatorial optimization problems with an application to multigenome
  alignment}, Operations Research \textbf{61} (2013), 453--468,
  \doi{10.1287/opre.1120.1139}.

\bibitem[{Moscato et~al.(2005)Moscato, Mathieson, Mendes, and
  Berretta}]{MMMB05}
P.~Moscato, L.~Mathieson, A.~Mendes, and R.~Berretta, \emph{The electronic
  primaries: Predicting the {U.\,S.} presidency using feature selection with
  safe data reduction}, V.~Estivill-Castro (ed.), \emph{ACSC 2005}, ACS,
  \emph{CRPIT}, vol.~38, 2005, pp. 371--380.

\bibitem[{Moser(2010)}]{Mos10}
H.~Moser, \emph{Finding Optimal Solutions for Covering and Matching Problems},
  Cuvillier, Göttingen, Germany, 2010.

\bibitem[{Niedermeier and Rossmanith(2003)}]{NR03}
R.~Niedermeier and P.~Rossmanith, \emph{An efficient fixed-parameter algorithm
  for {3-Hitting Set}}, Journal of Discrete Algorithms \textbf{1} (2003),
  89--102, \doi{10.1016/S1570-8667(03)00009-1}.

\bibitem[{O'Callahan and Choi(2003)}]{OC03}
R.~O'Callahan and J.-D. Choi, \emph{Hybrid dynamic data race detection},
  R.~Eigenmann and M.~Rinard (eds.), \emph{PPoPP'03}, ACM, 2003, pp. 167--178,
  \doi{10.1145/781498.781528}.

\bibitem[{Reiter(1987)}]{Rei87}
R.~Reiter, \emph{A theory of diagnosis from first principles}, Artificial
  Intelligence \textbf{32} (1987), 57--95, \doi{10.1016/0004-3702(87)90062-2}.

\bibitem[{Sorge et~al.(2014)Sorge, Moser, Niedermeier, and Weller}]{SMNW14}
M.~Sorge, H.~Moser, R.~Niedermeier, and M.~Weller, \emph{Exploiting a
  hypergraph model for finding {Golomb} rulers}, Acta Informatica \textbf{51}
  (2014), 449--471, \doi{10.1007/s00236-014-0202-1}.

\bibitem[{Sorge and Weller(2019)}]{SW19}
M.~Sorge and M.~Weller, \emph{The graph parameter hierarchy}, 2019,
  \urlprefix\url{https://manyu.pro/assets/parameter-hierarchy.pdf}.

\bibitem[{Vazquez(2009)}]{Vaz09}
A.~Vazquez, \emph{Optimal drug combinations and minimal hitting sets}, {BMC}
  Systems Biology \textbf{3} (2009), 81, \doi{10.1186/1752-0509-3-81}.

\bibitem[{Weihe(1998)}]{Wei98}
K.~Weihe, \emph{Covering trains by stations or the power of data reduction},
  R.~Battiti and A.~A. Bertossi (eds.), \emph{Proceedings of Algorithms and
  Experiments (ALEX 1998)}, 1998, pp. 1--8.

\end{thebibliography}

\appendix
\section{Error in the data reduction of \citet{MMMB05}}
\label{apx:counterexample}
\noindent
\looseness=-1
\citet[Section 2.2, Rule 2]{MMMB05} attempt to generalize \ref{w2} to Multiple Hitting Set
using the following data reduction rule:
if there are two vertices $v_i$, $v_j$ such that $E(v_j) \subseteq E(v_i)$, and for every $e \in E(v_j)$ one has
$|e|> f(e)$, then delete $v_j$ can be deleted.
We now show that this rule, in fact,
can change the cardinality of an optimal solution.

Consider the following hypergraph
with three hyperedges $e_1=\{v_1,v_2\}$, $e_2=\{v_2,v_3,v_4\}$, and $e_3=\{v_2,v_3,v_5\}$
with equal demands $f(e_1)=f(e_2)=f(e_3)=2$.
It has a multiple hitting set $\{v_1,v_2,v_3\}$,
yet the rule suggested by \citet{MMMB05}
could delete~$v_3$ due to~$v_2$.
After deletion of~$v_2$, however,
the minimum multiple hitting set is~$\{v_1,v_2,v_4,v_5\}$
and has cardinality four.

\end{document}